\documentclass[12pt,draftcls,onecolumn]{IEEEtran}

\usepackage{graphicx}
\usepackage{amsmath,graphicx,float,cite,dsfont,amssymb}
\usepackage{amsmath,amssymb,amsthm}
\usepackage{amsfonts,bm} 
\usepackage{mathtools}
\usepackage{color}
\usepackage[inline, shortlabels]{enumitem}
\usepackage{pgfplots}
\newtheorem{theorem}{Theorem}

\newtheorem{algorithm}{Algorithm}
\newtheorem{lemma}{Lemma}

\let\emptyset\varnothing
\DeclarePairedDelimiter\normV{\lVert}{\rVert}

\usepackage{textcomp}
\usepackage{algorithm}
\usepackage{algorithmic}
\usepackage{caption}
\usepackage{subcaption}
\usepackage{tikz}
\usepackage{graphicx}
\usepackage{balance}
\usepackage[makeroom]{cancel}
\usepackage{colortbl}
\theoremstyle{remark}

\let\emptyset\varnothing

\theoremstyle{definition}

\begin{document}

\title{Distributed Resource Management in Downlink Cache-Enabled Multi-Cloud Radio Access Networks}

\author{\IEEEauthorblockN{Robert-Jeron Reifert, \IEEEmembership{Student Member, IEEE},
		Alaa Alameer Ahmad, \IEEEmembership{Member, IEEE},
		Hayssam Dahrouj, \IEEEmembership{Senior Member, IEEE},
		Anas Chaaban, \IEEEmembership{Senior Member, IEEE},
		Aydin Sezgin, \IEEEmembership{Senior Member, IEEE},
		Tareq Y. Al-Naffouri, \IEEEmembership{Senior Member, IEEE},  and
		Mohamed-Slim Alouini, \IEEEmembership{Fellow, IEEE}}
	\thanks{Part of this paper was presented at the IEEE International Conference on Communications Workshops, June 2020 \cite{ICC}. \newline
	\IEEEauthorblockA{Robert-Jeron Reifert, Alaa Alameer Ahmad, and Aydin Sezgin are with Digital Communication Systems, Ruhr-Universit\"at Bochum, Bochum, Germany. (Email: \{robert-.reifert,alaa.alameerahmad,aydin.sezgin\}@rub.de)}\newline
	Hayssam Dahrouj, Tareq Y. Al-Naffouri and Mohamed-Slim Alouini are with King Abdullah University of Science and Technology, Thuwal, Saudi Arabia. (Email: \{hayssam.dahrouj,tareq.alnaffouri,slim.alouini\}@kaust.edu.sa)\newline
	Anas Chaaban is with the School of Engineering, The University of British Columbia, Kelowna, Canada. (Email: anas.chaaban@ubc.ca)
	}	
}
\maketitle
\begin{abstract}
	In light of the premises of beyond fifth generation (B5G) networks, the need for better exploiting the capabilities of cloud-enabled networks arises, so as to cope with the large-scale interference resulting from the massive increase of data-hungry systems. A compound of several clouds, jointly managing inter-cloud and intra-cloud interference, constitutes a practical solution to account for the requirements of B5G networks. This paper considers a multi-cloud radio access network model (MC-RAN), where each cloud is connected to a distinct set of base stations (BSs) via limited capacity fronthaul links. The BSs are equipped with local cache storage and baseband processing capabilities, as a means to alleviate the fronthaul congestion problem. The paper then investigates the problem of jointly assigning users to clouds and determining their beamforming vectors so as to maximize the network-wide energy efficiency (EE) subject to fronthaul capacity and transmit power constraints. This paper solves such a mixed discrete-continuous, non-convex optimization problem using fractional programming (FP) and successive inner-convex approximation (SICA) techniques to deal with the non-convexity of the continuous part of the problem, and $l_0$-norm approximation to account for the binary association part. A highlight of the proposed algorithm is its capability of being implemented in a distributed fashion across the network's multiple clouds through a reasonable amount of information exchange. The numerical simulations illustrate the pronounced role the proposed algorithm plays in alleviating the interference of large-scale MC-RANs, especially in dense networks.	
\end{abstract}
\begin{IEEEkeywords}
	Multi-cloud radio access network, energy efficiency, successive inner convex approximation, fractional programming, Dinkelbach algorithm, distributed implementation.
\end{IEEEkeywords}
\IEEEpeerreviewmaketitle

\thispagestyle{empty}

\section{Introduction}
\subsection{Overview}
\vspace{-.1cm}
Beyond fifth generation (B5G) wireless communication networks are expected to enable ultra-connectivity through the empowerment of Internet of Things (IoT) systems \cite{8820755}. 
IoT systems introduce unprecedented amounts of data traffic, thanks to the tremendous increase in the number of efficient mobile communication devices such as smartphones and tablets, and the extreme popularity of content-provider social media platforms such as YouTube and Netflix \cite{8820755,7397856,8088528}. Video data traffic leads to an exponential increase in mobile data traffic. The video data usage is anticipated to increase from 63\% of the total data traffic of 38 exabytes (EB) per month in 2019 to 76\% of the total data traffic of 160 EB per month in 2025 \cite{erricson1}. While the data traffic exponentially increases and the requirements for modern communication systems introduce new challenges, restraining the network's total energy consumption is vital. In recent years, cloud radio access networks (C-RANs) have emerged as promising network architecture to accommodate the requirements of B5G wireless networks. In C-RAN, a set of geographically distributed base stations (BSs) are connected to a central processor (CP) at the cloud via high-speed digital fronthaul links, which helps managing the exacerbating large-scale wireless interference \cite{quek_peng_simeone_yu_2017}. In conjunction with C-RANs development, edge caching in wireless networks promises nowadays to be an efficient technique to reduce the backhaul congestion in the network, and consequently to minimize the content delivery time during peak-traffic communication \cite{6600983}. This is achieved by storing the popular content at the BSs closer to end-users. Such caching approach can further improve the content delivery rate and reduce the communication latency via alleviating the communication load on the fronthaul links, which represents the bottleneck in achieving high data-rates in C-RANs.

The majority of works on C-RAN consider a single-cloud scenario, where a single CP in the cloud is responsible for coordinating the operation of the well-spread multi-cell networks, containing a large number of BSs and users (see \cite{DaiY14,7880686} and references therein). However, the plurality and widespread of devices in next-generation systems would necessitate the deployment of multiple CPs, each responsible for managing a distinct set of BSs \cite{6799231,7841744,7086838,e22060668}. Each CP at the cloud coordinates the data processing and beamforming vectors of the set of BSs associated with it. Such coordination between CPs, however, needs not to exacerbate the communication backbones and is rather limited to message passing among the different clouds; hence the need to distributively manage their underlying infrastructures on a message-passing level, which this paper tackles in details. We refer to a C-RAN with multiple CPs as multi-cloud radio access network (MC-RAN) to distinguish it from the classical single CP C-RAN.
In MC-RAN, the \textit{inter-cloud} interference becomes an additional performance barrier metric, especially given the limited communication between the distributed CPs; thus the need to jointly managing both the inter-cloud and the intra-cloud interference. Even more pronounced than in single CP C-RAN, backhaul congestion is a major problem in MC-RAN, especially during high-traffic times, due to the accentuated number of constrainted cloud-BS connections. Edge caching in MC-RANs becomes, therefore, vital for tackling the congestion of the backhaul link and, consequently, for minimizing delivery delay times.

{In this paper, we consider a content-based MC-RAN, where each cloud coordinates the operation of a set of BSs. Each BS is equipped with a local memory that has a certain storage capacity to cache the most popular files. The paper adopts the problem of maximizing the energy efficiency (EE) metric, which is defined as the rate-to-power ratio, so as to strike a good trade-off between the achievable rate and transmit power. In the MC-RAN model addressed in this paper, the system performance becomes a function of the user-to-cloud association and caching strategies, as well as the beamforming vector of each user. The paper tackles such a challenging non-convex mixed-integer optimization problem through an efficient algorithm that can be implemented in a distributed fashion across the multiple CPs}.
\subsection{Related Works}
The MC-RAN resource management problem considered in this paper is related to recent works on wireless edge caching, C-RANs, and distributed resource allocation.
The C-RAN architecture has been studied by several recent works, e.g., \cite{8732995,9500732,9445019,weinberger2021synergistic,8815579}. All these works deal with the efficient design and optimization of the resource allocation strategy in the considered C-RAN system architecture. For instance, \cite{8732995} utilizes the concept of rate-splitting and common message decoding to enhance the network-wide sum-rate. Work \cite{9500732} additionally enhances the C-RAN's performance by exploiting the characteristics of reflecting surfaces. However, these works do not consider multiple CPs nor do they exploit the caching at edge capabilities.
Wireless edge caching has received substantial attention in research communities recently. The seminal paper \cite{6600983} highlights the benefits of caching in reducing the end-to-end transmission delay and alleviating the bottleneck of fronthaul congestion in wireless communication. In \cite{8291028}, the authors investigate both coded and uncoded caching strategies and analyze their impact on the EE of the system. %
Cooperative caching and optimizing transmission schemes jointly in small cell networks are studied in \cite{7880694}. 
Concerned with the EE in cache-enabled networks, the work \cite{8374907} considers scalable video coding-based random and fractional caching in a single user, under a multiple BSs setting. To reduce energy consumption, latency, and enhance the cache hits, work \cite{9137229} considers an information centric network under a central control caching scheme in a network with multiple autonomous systems.
In \cite{9291438}, the authors consider a multi-cell multi-antenna fog-RAN with multiple cache instances and a recommendation scheme, which further boosts cache hits. A complex mixed-timescale joint recommendation, caching, and beamforming design problem to decrease the content transmission delay is formulated and solved centrally.
The seminal work \cite{7488289} considers a cache-enabled single CP C-RAN. The authors in \cite{7488289} investigate the dynamic content-centric BS clustering and multicast beamforming design and formulate the problem of minimizing a weighted sum of fronthaul cost and transmit power under the quality-of-service (QoS) constraints. The authors in \cite{Yigit} study a coded-caching strategy in C-RAN and use semi-definite relaxations (SDP) to optimize the beamforming vectors from BSs to users. %
Other works have studied edge caching in C-RAN with different objectives. In \cite{8269405}, the authors study the impact of caching on balancing the outage probability against fronthaul usage in a single CP C-RAN. The paper \cite{8269405} suggests a caching strategy that jointly optimizes the cell average outage probability and fronthaul usage. The paper \cite{7558153} studies the joint design of cloud and edge processing, where the edge nodes, i.e., the BSs, are equipped with local caches.
All these works \cite{7488289,Yigit, 8269405,7558153}, however, consider a single CP C-RAN model. References \cite{7488289,Yigit, 8269405}, in particular, consider a conventional C-RAN model in which the CP is responsible for performing most tasks in the baseband processing protocol, while radio transmission is done by the BSs. Reference \cite{7558153}, on the other hand, illustrates the necessity of the existence of baseband processing capabilities at the BSs. Such capabilities are necessary when the BSs are equipped with a local cache that enables them to send the content directly to the end-user without the need for CP interaction; thereby reducing the usage of fronthaul capacity. Promising results show that edge computing architecture has strong merits for meeting B5G system requirements, mainly, those related to enhanced mobile broadband (eMBB), massive machine-type communications (mMTC), and ultra-reliable low latency communications (URLLC) services \cite{DBLP}. For example, in \cite{antonio}, the authors propose a computational cost model, which directly links the resource blocks reserved for a certain service with the computational capacity required for performing the processing tasks. Moreover, it is shown in \cite{antonio,3Gpp} that the required resource blocks mainly depend on the type of service requested by the users. In \cite{8403960}, the authors utilize power and subchannel allocation schemes as well as edge caching. Methods to optimize the resource allocation in mobile edge computing networks are proposed in \cite{9069928}. A hybrid resource allocation algorithm for mobile edge computing networks is proposed in \cite{8647539}, which shows the algorithm's suitability with future IoT applications.

Recently, multi-cloud systems are studied in references \cite{6799231,7841744}, under the assumption that each CP adopts a compression-based transmission strategy. In the current article, however, we focus on the \textit{data-sharing} strategy, which can achieve better performance in terms of sum-rate \cite{7809154}. In the same direction, references \cite{7086838,8027131} consider the user-to-CP association problem in a multi-cloud setup. While reference \cite{7086838} assumes fixed beamforming and an infinite fronthaul capacity, reference \cite{8027131} partially overcomes this issue by assuming a discrete set of fixed resources associated with each cluster of BSs connected to a specific CP.

In summary, most of the previous works focus on either a single-cloud architecture, or consider the edge caching problem or the local processing power, but not the connection of both facets. In this article, however, we consider the downlink of an MC-RAN in which the BSs are equipped with local caches and baseband processing capabilities. The performance of such a system is a function of the user-to-cloud association and the baseband functional split between the CPs and the local BSs. As the caches require processing power, additional energy consumption at the BSs has to be considered \cite{7032101}. We propose a transmission scheme where the content requested by each user can be served directly from the BS, if it is stored in the cache, or can be retrieved from the CP in case local processing is not affordable. To the best of our knowledge, this is the first work that investigates the connection between edge caching and functional split in MC-RAN.

\subsection{Contributions}
Unlike the aforementioned references, this paper focuses on an MC-RAN setup and considers the problem of jointly determining the user-to-cloud association and the users' beamforming vectors by maximizing the EE subject to exclusive local or global processing constraints, per-BS power constraints, and per-BS fronthaul constraints. To tackle such a difficult mixed discrete-continuous non-convex optimization problem, we propose a solution that is based on fractional programming and successive inner-convex approximations (SICA) framework to determine the continuous variables, and a $l_0$-norm heuristic approximation to cope with the discrete (binary) variables. {A highlight of the proposed algorithm is its ability to determine the user-to-cloud association and beamforming vectors in a distributed fashion across the multiple clouds, which makes it amenable to practical implementation}.
Unlike our conference version \cite{ICC} which excludes the edge caching capabilities and focuses on the sum-rate maximization problem, this paper rather considers a cache-enabled MC-RAN and tackles the EE objective, which strikes a trade-off between achieving a reasonably high sum-rate for a relatively low power consumption. We herein consider a practical system model in which multiple CPs are responsible to manage a dense set of BSs, each equipped with local cache storage and baseband processing capabilities, as a means to alleviate the congestion of the fronthaul links across the multiple clouds of the network. The major contributions of this paper can be summarized as follows:
\begin{itemize}
	\item[1)] {\textit{Hybrid Transmission Scheme:}} 
	In the studied system model, we propose a flexible functional split between the CPs at clouds and the BSs. That is, if a BS caches the requested content, the baseband processing functions can be performed either locally at the BS, bypassing the interaction with the CP and the corresponding load on the fronthaul links, or centrally at the CP. Each functional split option determines a trade-off between the computation and fronthaul communication costs and results with different EE values.
	\item[2)] \textit{Optimization Framework:} We formulate an EE maximization problem subject to exclusive local or global processing constraints, per-BS power constraints, and per-BS fronthaul constraints. We then develop a general solution to the formulated non-convex problem, by first solving for the user-to-cloud association, then relaxing the binary variables using $l_0$-norm approximation, and finally solving for the continuous non-convex optimization problem using Dinkelbach-transform \cite{8187084}, with a SICA framework.
	\item[3)] \textit{Numerical Simulations:} We perform extensive numerical simulations to evaluate the performance of distributed and centralized implementations of the proposed scheme. We also compare these implementations against state-of-the-art schemes. In particular, we investigate the impact of the system parameters, i.e., cache size, processing costs, fronthaul capacity, and number of users, on the EE of the considered MC-RAN. We also illustrate the convergence behavior of the proposed algorithm for different system parameters.
\end{itemize}
\subsection{Notations}
Throughout the paper, boldface lower-case and capital letters (e.g. $\mathbf{h}$, $\mathbf{H}$) denote vectors and matrices, respectively. Calligraphic letters (e.g. $\mathcal{H}$) represent sets. A column vector consisting of all the elements in set $\mathcal{H}$ is defined as $ {\rm vec}(\mathcal{H}) $. If $\mathcal{H} = \{h_1,\cdots,h_N\}$, then $ {\rm vec}(\mathcal{H}) \equiv [h_1,\cdots,h_N]^T $. If $\mathcal{H} = \{\mathbf{h}_1,\cdots,\mathbf{h}_N\}$, then $ {\rm vec}(\mathcal{H}) \equiv [\mathbf{h}_1^T,\cdots,\mathbf{h}_N^T]^T $. $\mathbf{0}_N$ is a vector of length $N$ with all elements set to zero. The real and complex field are noted as $\mathbb{R}$ and $\mathbb{C}$, respectively, while the real part of complex numbers is given by $\Re\{\cdot\}$. Finally $\left( \cdot \right)^{\dagger}$ denotes the hermitian transpose and $\left( \cdot \right)^{T}$ the transpose operator, also $\left|\cdot\right|$ is the absolute value and $\normV[\big]{\cdot}_{p}$ the $l_p$-norm. 
\section{System Model}


\subsection{Received Signal Model}
\begin{figure}
	\begin{center}
		\includegraphics[scale=1.3]{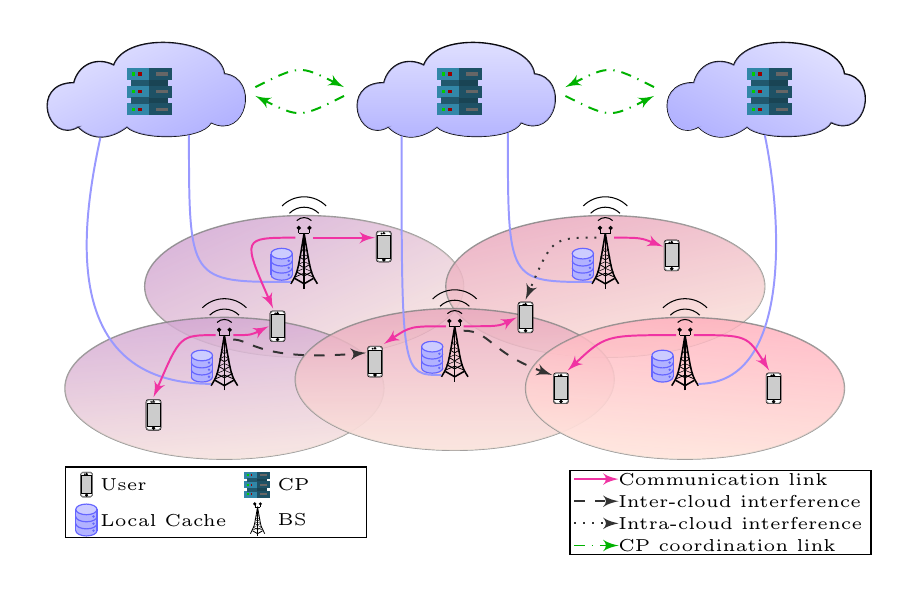} 
	\end{center}
	\caption{System model of an MC-RAN consisting of three clouds, $5$ BSs and $8$ users. Examples for inter-cloud and intra-cloud interference, as well as generic communication links are provided.} \label{sysmod}
	\vspace{-0.7cm}
\end{figure}
Consider the downlink of an MC-RAN 
 consisting of $C$ CPs, where each CP $c$ coordinates a set of BSs of size $B_c$. Each BS is assumed to have $L$ antennas, connected to one (and only one) CP via a digital fronthaul link with finite capacity, and equipped with a local cache memory of a total of $F_b \leq F$ local files, where $F$ denotes the total number of the library files. The network consists of $K$ single-antenna users. Fig.~\ref{sysmod} illustrates an example of the considered model, with an MC-RAN of $3$ clouds, and a total of $5$ BSs and $8$ users.

Let $\mathcal{C}=\{1,\cdots,C\}$ be the set of CPs, and $\mathcal{B}=\{1,\cdots,B\}$ be the set of BSs in the network, where $B = \sum_{c\in\mathcal{C}}B_c$. Furthermore, let $\mathcal{K}=\{1,\cdots,K\}$ be the set of users, and $\mathcal{F}=\{1,\cdots,F\}$ be the set of all files. We assume that each user $k\in \mathcal{K}$ can be assigned to one and only one CP $c \in \mathcal{C}$. Furthermore, we assume that every CP $c \in \mathcal{C}$ is connected to a cluster of BSs denoted by $\mathcal{B}_c = \{1,\cdots,B_c\}$. The networks clusters $\mathcal{B}_c$, $c\in\mathcal{C}$ are assumed to be disjoint, i.e., $\cup_{c \in \mathcal{C}} \mathcal{B}_c = \mathcal{B}$, $\mathcal{B}_c \cap \mathcal{B}_{c'}  = \emptyset, \forall c\neq c'$.

Let $\mathbf{h}_{c,b,k}\in {\mathbb C^{L}}$ be the channel vector from the $b$-th BS of the $c$-th cloud to the $k$-th user, and let ${\mathbf{h}}_{c, k}\in {\mathbb C}^{B_c L \times 1}$ be the aggregated channel vector from the $c$-th cloud to the $k$-th user. This can be expressed as ${\mathbf{h}}_{c,k}\triangleq [\mathbf{h}_{c,1,k}^{T},\cdots, \mathbf{h}_{c,B_c,k}^{T}]^T$. To simplify our discussion and make the problem mathematically tractable, we assume that each CP has access to the full channel state information (CSI), to the cached content of BSs in $\mathcal{B}_c$ and the demands (requested files) of all the users in the network. To deliver the requested files, we adopt a time-slotted block-based transmission model where each transmission block consists of several time slots. The channel fading coefficients remain constant within one block and may vary independently from one block to another. We focus on optimizing the EE of the cache-aided MC-RAN within one transmission block. Without loss of generality, we consider that the CPs divide each requested file into several data chunks, so that the transmission of each file may take place on several consecutive transmission blocks and the number of required transmission blocks to transmit each file may be different from other files.

\subsection{Cache Model}
In content delivery networks (CDN), edge caching is employed to bring the content closer to users. In general, we can distinguish between two phases in content delivery process to mobile users, namely \textit{cache placement} and \textit{cache delivery} phases \cite{8008769,6763007 }. Therefore, the recent works studying cache-aided wireless networks can be divided into two main categories: 1) optimizing the cached content delivery process for a given cache placement to get the best possible performance \cite{6336688}; 2) improve the content delivery process through efficient design of cache placement strategies \cite{8269405}. Essentially, in the cache placement phase, the popular content is stored at edge-network nodes, i.e., at BSs, with the sole purpose of improving the cache delivery phase, especially in peak-traffic times. Hence, the cache placement phase takes place over a much longer time-scale than that required in the cache delivery phase, since the popularity of the content changes much slower than the time required to deliver the requested content to the users.

In this article, we focus on the optimization of the cache delivery phase, while the cache placement phase is considered to be performed a priori. The cache content at BSs, the user's requests, and the whole library of files are assumed to be known at the clouds.
Let $\mathbf{C} \in \left\lbrace 0,1\right\rbrace^{F \times B} $ be the binary cache placement matrix where $\left[\mathbf{C}\right]_{f, b} = c_{f,b} $ is the element in the $f$-th row and the $b$-th column. Now, let $f_k \in \mathcal{F}$ be the requested file of user $k$, then $c_{f_k,b} = 1 $, a \textit{cache hit}, means that $f_k$ is cached at BS $b$ and $c_{f_k,b} = 0 $, a \textit{cache miss}, means it is not. Define the set of \textit{cache hit} users as $\mathcal{K}_1 \triangleq\left\lbrace k \in\mathcal{K}| \hspace{1mm} \exists (c, b)\in \mathcal{C} \times \mathcal{B}: c_{f_k, b} = 1 \right\rbrace $. 
Hence, the set $\mathcal{K}_1$ contains all users whose requested files are cached locally at the BSs. On the other hand, we define the set of \textit{cache-miss} users $\mathcal{K}_2$ as the set of users whose requested files are not cached the BSs, i.e., $\mathcal{K}_2 \triangleq\left\lbrace k \in\mathcal{K}| \hspace{1mm} \forall (c, b)\in \mathcal{C} \times \mathcal{B}: c_{f_k, b} = 0 \right\rbrace $. Note that in the special case where $\mathbf{C} = \mathbf{0}_{F\times B}$, no files are stored in BSs caches, i.e., $\mathcal{K}_1= \emptyset $ and $\mathcal{K}_2= \mathcal{K}$. On the other hand, when $\sum_{b \in \mathcal{B}}c_{f,b} \geq 1 \hspace{1mm} \forall f \in \mathcal{F}$, each file is cached at least at one BS, i.e., $\mathcal{K}_2= \emptyset $ and $\mathcal{K}_1= \mathcal{K}$.
\subsection{Baseband Processing and Fronthaul Communication Cost Models}

The majority of works on wireless caching consider the fronthaul communication and transmit costs, but ignore the computation cost required to process the requested file before transmitting it to the users. Unlike previous works, in this paper, we account for such factors while optimizing the delivery phase strategy such that the EE of the MC-RAN is maximized. The paper considers a computational cost model, in which the processing cost associated with each requested content is fixed and depends on the service requested by the user and the number of resource blocks served for delivering the requested content. We denote this cost as $\left\lbrace P_{k}^{\text{proc}}, \, \forall k \in \mathcal{K}\right\rbrace $.
Next, we describe different transmitting strategies considered in the cache-aided MC-RAN (i.e., at CPs and BSs). Each uses the processing resources, fronthaul links, and transmit resources differently.
\section{Transmit Scheme and Functional Split}
\subsection{Design Transmit Signals at the CP}
In this paper, we focus on the \textit{data-sharing} transmission strategy. In this strategy, the CP performs joint encoding of users' data. 
In more details, CP $c$ encodes $v_k$, the data chunk of the file $f_k$ requested by user $k$, into $s_k$. Here $s_k$ denotes the symbol of the encoded data at CP $c$ to be transmitted to user $k$ at the current time slot. We assume that $s_k$ is chosen independently from a complex Gaussian distribution with zero-mean and unit variance. After that, the CP forwards $s_k$, the encoded data chunks, through limited capacity fronthaul links to the cluster of BSs serving user $k$. The BSs then cooperate to transmit the signal to user $k$ using a joint beamforming vector. Although the beamforming vector coefficients are designed at the CP, the modulation and precoding tasks are performed at the BSs. We assume that the rate required to transmit beamforming vector coefficients over the fronthaul links is negligible compared to that required for transmitting the coded symbols of the users \cite{DaiY14}.
Let $\mathbf{w}_{c,k} \in \mathbb{C}^{B_c L \times 1} = \left[\mathbf{w}_{c,1,k}^{T},\ldots, \mathbf{w}_{c,B_c,k}^{T}\right]^T$ be the aggregate beamforming vector of user $k$ when associated with CP $c$.
Note that if BS $b \in \mathcal{B}_c$ is not in the BSs' cluster serving user $k$, then $\mathbf{w}_{c,b,k} = \mathbf{0}_L$, and the CP in this case does not share any data of user $k$ with BS $b$. Thus, the aggregate beamforming vector $\mathbf{w}_{c, k}$ is a group-sparse vector by construction.

\subsection{Design Transmit Signals Locally at BSs}
Caching the most popular files locally at the BSs overcomes the disadvantages of processing the data at the cloud and significantly reduces the load on the fronthaul links. Hence, the usage of the fronthaul link boils down to the exchange of essential control information between CPs and BSs (e.g., beamforming coefficients and scheduling information). However, despite the advantages of caching the content locally at BSs in terms of reducing latency and saving the fronthaul link bandwidth, this comes at the cost of increasing the processing cost at the BSs.
Hence, we consider that the BSs cache the uncoded data locally, and so encoding the data before transmission is done at the local processing unit at the BSs. That is, we assume the baseband processing tasks are done at the BSs serving user $k$, whenever the BSs cache the required file of user $k$. To this end, let $\mathbf{\tilde{w}}_{c, b, k} \in \mathbb{C}^{L \times 1}$ be the beamforming vector explicitly used at BS $b \in \mathcal{B}_c$ for user $k$ when the baseband processing tasks are performed at the BS. Let $\mathbf{\tilde{w}}_{c,k} \in \mathbb{C}^{B_c L \times 1} = \left[\mathbf{\tilde{w}}_{c,1,k}^{T},\ldots, \mathbf{\tilde{w}}_{c,B_c,k}^{T}\right]^T$ be the aggregate beamforming vector at BSs in cluster $\mathcal{B}_c$ which cache the requested file of user $k$. %
Note that the BS $b$ can encode the data locally and independently of the CP connected to it when it caches the requested file from user $k$. The control information needed for transmitting the signal is, however, assumed to be provided from the cloud. Hence, partial cooperation between CPs on the control level becomes possible.
\subsection{Hybrid Transmit Strategy}
This paper considers a hybrid transmit strategy, where each BS $b$ can serve a user $k$ (for each user $k \in\mathcal{K}$) according to one of three possibilities: 1) The BS can participate in transmitting data to $k$ following the CP processing strategy, i.e., based on the CP encoding, 2) Or, the BS processes the data locally, i.e., in caching scenarios, 3) Or, the BS does not transmit to user $k$. The transmit signal at BS $b$ from cluster $\mathcal{B}_c$, $\mathbf{x}_{c, b} \in \mathbb{C}^{L \times 1}$, can thus be written as follows
\begin{equation}\label{eq:BSsig}
\mathbf{x}_{c, b} = \sum_{k \in \mathcal{K}} \left( \mathbf{w}_{c, b, k} + \mathbf{\tilde{w}}_{c, b,k} \right)  s_k.
\end{equation}
The encoding process can be either done at the cloud or locally at the BS, but not at both at the same time. Also, the BS can perform the processing locally only in case it caches the requested file. Therefore, equation \eqref{eq:BSsig} is accompanied by the following two conditions:
\begin{align}
\mathds{1}\left\lbrace{\normV[\big]{\mathbf{w}_{c,b, k}}_{2}^{2}} \right\rbrace + \mathds{1}\left\lbrace{\normV[\big]{\mathbf{\tilde{w}}_{c, b, k}}_{2}^{2}} \right\rbrace &\leq 1,\label{eq:BME} &&\forall k \in \mathcal{K}, \forall b \in \mathcal{B}_c, \\
{\mathbf{\tilde{w}}_{c, b, k}} & = \mathbf{0}_L, &&\forall k \in \mathcal{K}_2, \forall b \in \mathcal{B}_c, \label{eq:BM}
\end{align}
where $\mathds{1}\left\lbrace{\cdot}\right\rbrace$ is the indicator function defined as:  $\mathds{1}\left\lbrace{x} \right\rbrace = 1$  if $x > 0$, and $0$ otherwise.
Equations \eqref{eq:BSsig}, \eqref{eq:BME}, and \eqref{eq:BM} can be interpreted as follows: if file $f_k$ is not cached at BS $b$ then BS $b \in \mathcal{B}_c$ transmits to user $k$ only if $\mathds{1}\left\lbrace{\normV[\big]{\mathbf{w}_{c,b, k}}_{2}^{2}} \right\rbrace = 1$, since in this case ${\mathbf{\tilde{w}}_{c, b, k}}  = \mathbf{0}_L$ . In case file $f_k$ is cached at BS $b$, the data would be encoded at the CP or locally at the BS depending on equation \eqref{eq:BME}. Moreover, by construction, if user $k$ is not associated with CP $c$ then $\mathbf{w}_{c,b, k} = \mathbf{0}, \forall b \in \mathcal{B}_c$. The specific design of beamforming vectors $\mathbf{w}_{c,b, k}$ or $\mathbf{\tilde{w}}_{c, b, k}$ in this case is based on solving our optimization problem, as discussed in details in section IV. After forming the transmit signal as in (\ref{eq:BSsig}), BS $b$ sends $\mathbf{x}_{c, b}$ subject to the following maximum transmit power constraint:
\begin{equation}
\mathbb{E}\left\lbrace\mathbf{{x}}_{c,b}^{\dagger}\mathbf{{x}}_{c,b} \right\rbrace \leq P_{b}^{\text{Max}}.
\end{equation}
\subsection{Achievable Rates and Fronthaul Constraints}
User $k$ can be served by one (and only one) BS-cluster connected to CP $c$ with an aggregate beamforming vector $\mathbf{w}_{c,k}$. Define the user-to-cloud association as a binary variable $z_{c,k}$, i.e., $z_{c,k} = 1$ if user $k$ is associated to cloud $c$, and $0$ otherwise. We further assume that each user can be associated to one and only one CP since, otherwise, a signal-level coordination would be required between the clouds, rather than a control-level coordination. 
We can now write the signal to interference plus noise ratio (SINR) of user $k$ when associated with CP $c$ as follows
\begin{equation}
\label{SINR_definitionH}
\mathrm{SINR}_{c,k} =\frac{\big|{\mathbf{h}}_{c, k}^\dagger\left( {\mathbf{w}_{c,k}}+{\mathbf{\tilde{w}}_{c,k}}\right) \big|^2}{\sum_{(c', k')\neq (c, k)}{\big|z_{c',k'}\,{\mathbf{h}}_{c', k}^\dagger\left( {\mathbf{w}_{c', k'}}+ {\mathbf{\tilde{w}}_{c', k'}}\right) }\big|^2 + \sigma^2} .
\end{equation}
Let $W$ be the bandwidth allocation of user $k$. The achievable rate of user $k$ associated to cloud $c$ becomes bounded as
\begin{equation}
R_{c, k} \leq W\log_2(1+\rm{SINR}_{c,k}).
\end{equation}
The transmit power per-BS can be expressed as
\begin{equation}\label{eq:transmitpowerH}
P_{b}\left(\mathbf{w}, \mathbf{\tilde{w}}\right) =  \frac{1}{\eta_b} \sum_{k \in \mathcal{K}} \left(\normV[\big]{\mathbf{w}_{c,b, k}}_{2}^{2} + \normV[\big]{\mathbf{\tilde{w}}_{c, b, k}}_{2}^{2}\right),
\end{equation}
where $\eta_b < 1 $ is the efficiency of the transmit amplifier at BS $b$.
The required fronthaul capacity at BS $b$ is given as
\begin{equation}\label{eq:Fronthl}
C_{b}\left(\mathbf{w}, \mathbf{\tilde{w}}\right) = \sum_{k \in \mathcal{K}}\left(\mathds{1}\left\lbrace{\normV[\big]{\mathbf{w}_{c,b, k}}_{2}^{2}} \right\rbrace +  (1-c_{f_k,b})\mathds{1}\left\lbrace{\normV[\big]{\mathbf{\tilde{w}}_{c, b, k}}_{2}^{2}} \right\rbrace \right) R_{c,k},
\end{equation}
where $\mathbf{w} \triangleq {\rm vec}(\left\lbrace \mathbf{w}_{c , k}| \forall (c, k) \in  \mathcal{C}\times \mathcal{K} \right\rbrace) $, $\mathbf{\tilde{w}} \triangleq {\rm vec}(\left\lbrace \mathbf{\tilde{w}}_{c , k}| \forall (c, k) \in  \mathcal{C}\times \mathcal{K}\right\rbrace) $. It is obvious from \eqref{eq:Fronthl} that if BS $b$ caches file $f_k$ requested by user $k$, i.e., $c_{f_k,b} = 1$, then user $k$ does not add to the burden of the fronthaul link of BS $b$ when $\mathbf{w}_{c,b, k} = \mathbf{0}_L$.
\subsection{Energy Efficiency at the Cloud}
In the context of our paper, the individual EE metric of each cloud $c$ is defined as the sum-rate of all users associated with $c$ divided by the power consumption required to serve these users. This work takes into account the transmit power, processing power, fronthaul power consumption, and operational fixed power consumption. The latter does not depend on the number of users and is defined as $P_c^{\text{Pr}}$. Such operational fixed power allocation include, but is not limited to, required cooling and circuitry power resources for the functionality of the C-RAN. Mathematically we define the EE at the cloud $c$ as follows%
%
\begin{equation}\label{ee}
	f_{\text{EE}}(c) \triangleq \frac{\sum\nolimits_{k \in \mathcal{K}}{R}_{c, k}}{P_c^{\text{Tx}}+g_{\text{EE}}(c)+P_c^{\text{Pr}}},
\end{equation}
where $P_c^{\text{Tx}}$ is the total transmit power consumed by the BSs of cluster $\mathcal{B}_c$ defined as
\begin{equation}
P_c^{\text{Tx}} = \sum_{b \in \mathcal{B}_c}P_{b}\left(\mathbf{w}, \mathbf{\tilde{w}}\right),
\end{equation}
where the processing power of the fronthaul, the BS, and the CP can be written as
\begin{equation}\label{hee}
	g_{\text{EE}}(c) = \underbrace{\sum\limits_{(b, k) \in \mathcal{B}_c \times \mathcal{K}} \mathds{1}\big\lbrace{\normV[\big]{\mathbf{{w}}_{c, b,k}}_{2}^{2}} \big\rbrace P_{b}^{\text{fthl}}}_{\text{Fronthaul processing cost}} + \underbrace{\sum\limits_{(b, k) \in \mathcal{B}_c \times \mathcal{K}} \mathds{1}\big\lbrace{\normV[\big]{\mathbf{\tilde{w}}_{c, b,k}}_{2}^{2}} \big\rbrace P_{k}^{\text{proc}}}_{\text{Processing at BS}} + \underbrace{\sum\limits_{k \in \mathcal{K}} \mathds{1}\big\lbrace{\normV[\big]{\mathbf{w}_{c,k}}_{2}^{2}} \big\rbrace P_{k}^{\text{proc}}}_{\text{Processing at CP}}.
\end{equation}
Interestingly, \eqref{ee} captures the trade-off between the local processing of cached files at the BSs and the fronthaul usage when the files are processed at the CP.
The next section tackles the paper's main optimization problem, which aims at maximizing the sum of EE of all the clouds, so as to determine the user-to-cloud assignment, the processing power decision variables, user-to-BS association, and the joint transmit beamforming vectors for all users across the network.
\section{MC-RANs EE Maximization and Algorithms}
This section first formulates the EE maximization problem as a mixed-integer optimization problem. To best tackle the intricacies of the problem at hand, the paper then presents some well-chosen mathematical reformulations, so as to derive an efficient iterative algorithm, the highlight of which is that it can be implemented in a distributed fashion across the network CPs.
\subsection{General Problem}
In the context of distributed EE across the MC-RAN, we seek to jointly optimize the functional split mode for each BS, the beamforming vectors and user-to-cloud association of all users in the network subject to per BS maximum transmit power and maximum fronthaul capacity constraints. The optimization problem under consideration can be mathematically written as:
\begin{subequations}\label{eq:Opt1}
	\begin{align}
	&\underset{\mathbf{w},\mathbf{\tilde{w}}, \mathbf{z},\mathbf{r}}{\text{maximize}}\quad \sum_{c \in \mathcal{C}} f_{\text{EE}}(c)  \label{eq:Obj} \\
	&\text{subject to} \quad \eqref{eq:BME},\eqref{eq:BM}, \nonumber\\
	& P_{b}\left(\mathbf{w}, \mathbf{\tilde{w}}\right)  \leq P_{b}^{\text{Max}} &&\forall b \in \mathcal{B}_c, \forall c \in \mathcal{C} \label{eq:powe},\\
	& C_{b}\left(\mathbf{w}, \mathbf{\tilde{w}}\right) \leq F_{b,c}  &&\forall b \in \mathcal{B}_c, \forall c \in \mathcal{C} \label{eq:FrontC},\\
	& \rm{SINR}_{c,k} \geq 2^{R_{c,k}/W} - 1  &&\forall k \in \mathcal{K}, \forall c \in \mathcal{C} \label{eq:SICn},\\
    & \sideset{}{_{c \in \mathcal{C}}}\sum z_{c,k} = 1 \quad &&\forall k\in \mathcal{K} \label{eq:Clds},\\	
    & z_{c,k} \in \left\lbrace 0, 1 \right\rbrace &&\forall k \in \mathcal{K}, \forall c \in \mathcal{C} \label{eq:bin1}, \\
    & \normV[\big]{\mathbf{w}_{c,b,k}}_{2}^{2} \leq M z_{c, k} &&\forall k \in \mathcal{K}, \forall b \in \mathcal{B}_c, \forall c \in \mathcal{C} \label{eq:BgM},\\
    & \normV[\big]{\mathbf{\tilde{w}}_{c,b,k}}_{2}^{2} \leq M z_{c, k} &&\forall k \in \mathcal{K}, \forall b \in \mathcal{B}_c, \forall c \in \mathcal{C} \label{eq:BgM2},\\
    & \sideset{}{_{k \in \mathcal{K}}}\sum z_{c,k} \leq K_c^{\text{Max}}  &&\forall c\in \mathcal{C}, \label{eq:Cap0}
 	\end{align}
\end{subequations}
where $\mathbf{z}  = {\rm vec}(\left\lbrace {z}_{c, k} | \forall \left(c, k\right) \in  \mathcal{C}\times \mathcal{K}\right\rbrace) $, $\mathbf{r}  = {\rm vec}(\left\lbrace {R}_{c, k} | \forall \left(c, k\right) \in  \mathcal{C}\times \mathcal{K}\right\rbrace)$, and $M$ is a large positive integer $M \in \mathbb{R}_{++}$ that is related to the big\textit{-M} constraint \cite{6514675}. $P_{b}^{\text{Max}}$ and $F_{b, c}$ are the maximum transmit power and the fronthaul capacity of BS $b$ in cloud $c$, respectively. $K_c^{\text{Max}}$ is the maximum number of users that can connect to cloud $c$. 
Constraint \eqref{eq:powe} represents the maximum transmit power available to BS $b$, and constraint \eqref{eq:FrontC} represents the available fronthaul capacity of BS $b$ connected to CP $c$. Constraint \eqref{eq:SICn} gives an upper bound on the maximum achievable rate of user $k$ when assigned to cloud $c$. Constraints \eqref{eq:Clds}-\eqref{eq:bin1} assure that each user can be associated with one and only one CP. Constraints \eqref{eq:BgM}-\eqref{eq:BgM2} represent the big\textit{-M} constraints and can be read as follows: if the CP $c$ is associated with user $k$, then constraints \eqref{eq:BgM}-\eqref{eq:BgM2} are deactivated \cite{6514675}. Otherwise, if $k$ is not associated with $c$, constraint \eqref{eq:BgM} forces the corresponding beamforming coefficients in $\mathbf{w}_{c,b,k}$ to zero, also \eqref{eq:BgM2} forces $\mathbf{\tilde{w}}_{c,b,k}$ to zero. The number of associated users to cloud $c$ does not exceed a given maximum number of users, which is ensured by \eqref{eq:Cap0}.


The above optimization \eqref{eq:Opt1} is over the binary variables $\mathbf{z}$, the continuous beamforming vectors $\mathbf{w}$ and $\mathbf{\tilde{w}}$, and the rates $\mathbf{r}$. Problem \eqref{eq:Opt1} is challenging to solve due to the non-convexity of the objective function and constraints \eqref{eq:FrontC}-\eqref{eq:bin1}, besides the discrete nature of variables $\mathbf{z}$.
\subsection{Overall Algorithmic Framework}
The optimization of the association variables and beamforming vectors in \eqref{eq:Opt1} is hard to tackle jointly and may be computationally prohibitive to solve globally. Therefore, our paper proposes adopting a two-step optimization approach. In the first step, we adopt an auxiliary, local utility function that represents the benefit of associating a user $k$ to cloud $c$. Then we formulate a generalized assignment problem to find the user-to-cloud association $\mathbf{z}$ for that given utility function. Afterwards, in the second stage, given the user-to-cloud assignment solution, we solve the optimization problem \eqref{eq:Opt1} using a $l_0$-norm relaxation followed by a successive inner-convex approximation approach to find the beamforming vectors $\mathbf{w}$ and $\mathbf{\tilde{w}}$, and the rates $\mathbf{r}$. We next present the generalized assignment formulation to solve the user-to-cloud association problem.
\subsection{User-to-Cloud Association}
Since problem \eqref{eq:Opt1} is too complicated to solve in full generality, we propose an ad-hoc solution to find the user-to-cloud association, which serves as a network planning step.
More precisely, let $\mathcal{U}(c,k)$ be an auxiliary, local utility function that measures the benefit of associating user $k$ with cloud $c$. A reasonable choice of $\mathcal{U}(c,k)$ is the following EE-like function:
\begin{equation} \label{utility}
\mathcal{U}(c,k) = \frac{R_{c, k}}{\sum_{b \in \mathcal{B}_c}\frac{1}{\eta_b} \left(\normV[\big]{\mathbf{w}_{c,b, k}}_{2}^{2} + \normV[\big]{\mathbf{\tilde{w}}_{c, b, k}}_{2}^{2}\right) + P_{k}^{\text{proc}}}.
\end{equation}
The intuition behind such choice is two-fold. Firstly, the utility function in \eqref{utility} defines the benefit of associating user $k$ with cloud $c$ as the ratio of the achievable rate for such an association and the processing and transmit power costs. Such a utility helps to mimic a reasonable, local EE expression of the MC-RAN, given that the utility function \eqref{utility} depends mainly on the aggregate beamforming vector from cloud $c$ to user $k$. Secondly, such choice helps to formulate a generalized assigned problem, which allows us to derive efficient algorithms to find the association variables $\mathbf{z}$; thereby alleviating the complexity of the solution of the complex problem \eqref{eq:Opt1}. The simulation results of the paper further validate the numerical gain of such heuristic approach. At this step, we fix the beamforming vectors from cloud $c$ to user $k$, e.g., $\mathbf{w}_{c,k}$ can have maximum ratio combining (MRC) structures defines as:
$\mathbf{w}_{c,k} \in \mathbb{C}^{B_c L\times 1} = \frac{\mathbf{h}_{c,k}}{\normV[]{\mathbf{h}_{c, k}}_{2}^2}, \quad \forall k \in \mathcal{K}$. Given the above utility $\mathcal{U}(c,k)$ and the fixed beamforming strategy, problem \eqref{eq:Opt1} boils down to a user-to-cloud association problem, which can be written as: \\
\begin{subequations}\label{eq:Opt2}
	\begin{align}
	&\underset{ \mathbf{z}}{\text{maximize}}\quad \sum_{(c, k) \in  \mathcal{C}\times \mathcal{K}} z_{c, k} \hspace{2mm}\mathcal{U}(c, k)  \label{eq:Obj1} \\
	&\text{subject to} \nonumber\\
	& \sideset{}{_{k \in \mathcal{K}}}\sum z_{c,k} \leq K_c^{\text{Max}}  &&\forall c\in \mathcal{C}, \label{eq:Cap1}\\
	& \sideset{}{_{c \in \mathcal{C}}}\sum z_{c,k} \leq 1 &&\forall k\in \mathcal{K}, \label{eq:Clds1}\\	
	& z_{c,k} \in \left\lbrace 0, 1 \right\rbrace &&\forall k\in \mathcal{K}, \forall c\in \mathcal{C}. \label{eq:bin2}
	\end{align}
\end{subequations}
Problem \eqref{eq:Opt2} follows a generalized assignment problem form \cite{6697043}, and is carried over the binary variables $\mathbf{z}$. While using global centralized optimization methods (e.g., the branch and cut algorithm) is possible, solving \eqref{eq:Opt2} in a distributed manner is rather adopted in the context of our MC-RAN setup, which is done using an auction-based iterative algorithm, similar to \cite{7086838} and \cite{6697043}, where only reasonable information exchange between the clouds is required. The algorithm guarantees convergence with a finite amount of iterations to a solution which is within a gap of $(1+\chi)$ to the global optimal solution of \eqref{eq:Opt2} \cite[Theorems 1 and 2]{6697043}, where $\chi \in [1,+\infty]$ is the approximation ratio of a subroutine knapsack algorithm.
\subsection{Problem Reformulation}
Fixing the user-to-cloud association solution $\mathbf{z}$ as described above, problem \eqref{eq:Opt1} can be reformulated as the following joint beamforming optimization problem:
\begin{subequations}\label{eq:Optbeam}
	\begin{align}
		&\underset{\mathbf{w}, \mathbf{\tilde{w}}, \mathbf{r}}{\text{maximize}}\quad \sum_{c \in \mathcal{C}} f_{\text{EE}}(c) \label{eq:Obj2_0} \\
		&\text{subject to} \quad \eqref{eq:BME},\eqref{eq:BM}, \eqref{eq:powe}, \eqref{eq:FrontC}, \eqref{eq:SICn}, \eqref{eq:BgM}, \eqref{eq:BgM2}. \nonumber
	\end{align}
\end{subequations}
While problem \eqref{eq:Optbeam} only runs over the continuous variables $\mathbf{w}$ and $\mathbf{\tilde{w}}$, and $\mathbf{r}$, it is still challenging due to the non-convexity of the objective function and feasible set. We next tackle such challenges by reformulating the non-convex constraints \eqref{eq:FrontC} and \eqref{eq:SICn}. We begin by reformulating problem \eqref{eq:Optbeam} as follows:
\begin{subequations}\label{eq:Opt3}
	\begin{align}
	&\underset{\mathbf{w}, \mathbf{\tilde{w}}, \bm {\gamma},\mathbf{r}}{\text{maximize}}\quad \sum_{c \in \mathcal{C}} f_{\text{EE}}(c) \label{eq:Obj2} \\
	&\text{subject to} \quad \eqref{eq:BME},\eqref{eq:BM}, \eqref{eq:powe}, \eqref{eq:FrontC}, \eqref{eq:BgM}, \eqref{eq:BgM2}, \nonumber\\
	& R_{c,k} \leq W\log_2\left(1+ \gamma_{c, k}\right) &&\forall k \in \mathcal{K}, \forall c \in \mathcal{C},  \label{eq:q1}\\
	& \sigma^2  + \sum_{(c', k')\neq (c, k)}{\big|{\mathbf{h}}_{c', k}^\dagger\left( {\mathbf{w}_{c', k'}}+ {\mathbf{\tilde{w}}_{c', k'}}\right) }\big|^2 - \frac{\big|{\mathbf{h}}_{c, k}^\dagger\left( {\mathbf{w}_{c,k}}+{\mathbf{\tilde{w}}_{c,k}}\right) \big|^2}{\gamma_{c, k}} \leq 0 \label{eq:q2} &&\forall k \in \mathcal{K}, \forall c \in \mathcal{C},
	\end{align}
\end{subequations}
where we introduce the variables $\bm \gamma = {\rm vec}(\left\lbrace \gamma_{c, k}| \forall \left(c, k\right) \in  \mathcal{C}\times \mathcal{K} \right\rbrace) $ to reformulate the maximum achievable rate constraint \eqref{eq:SICn} into constraints \eqref{eq:q1}-\eqref{eq:q2}. Constraint \eqref{eq:q2} is now in the form of difference of convex (DC) functions which can be tackled using an efficient SICA approach. Please note that the inter-cloud and intra-cloud interference terms (middle terms) in \eqref{eq:q2} do not contain the user-to-cloud association $z_{c',k'}$ anymore. Ensured by the  big-$M$ constraints \eqref{eq:BgM}-\eqref{eq:BgM2}, the beamforming vectors $\mathbf{w}_{c', k'}$ and $\mathbf{\tilde{w}}_{c', k'}$ are forced to zero if a user $k'$ is not associated with cloud $c'$. That is, the beamforming vectors implicitly include the, now fixed, user-to-cloud association variables. Therefor, the formulation \eqref{eq:q2} still holds.

Moreover, $\mathds{1}\big\lbrace{\normV[\big]{\mathbf{\tilde{w}}_{c, b,k}}_{2}^{2}} \big\rbrace$ indicates if the data of user $k$ is processed locally at BS $b$, $\mathds{1}\big\lbrace{\normV[\big]{\mathbf{{w}}_{c,k}}_{2}^{2}} \big\rbrace$ denotes if CP $c$ processes the data of user $k$, and $\mathds{1}\big\lbrace{\normV[\big]{\mathbf{{w}}_{c, b,k}}_{2}^{2}} \big\rbrace$ infers whether BS $b$ is in the serving cluster of user $k$ or not. These indicator functions present additional hurdles within the framework of the challenging problem \eqref{eq:Opt3}. We note that the benefit of using indicator functions is to determine the decision variables exclusively based on beamforming vectors. To deal with the challenging non-convex discrete indicator functions, we next use the $l_0$-norm relaxation technique.

First, we note that the indicator function in the objective \eqref{eq:Obj2} and the fronthaul constraint \eqref{eq:FrontC} can be equivalently expressed as an $l_0$-norm of the beamforming vectors. We can write $\mathds{1}\big\lbrace{\normV[\big]{\mathbf{\tilde{w}}_{c, b,k}}_{2}^{2}} \big\rbrace \triangleq \normV[\big]{\normV[\big]{\mathbf{\tilde{w}}_{c, b,k}}_{2}^{2}}_0 $, $\mathds{1}\big\lbrace{\normV[\big]{\mathbf{{w}}_{c, b,k}}_{2}^{2}} \big\rbrace \triangleq \normV[\big]{\normV[\big]{\mathbf{{w}}_{c, b,k}}_{2}^{2}}_0 $, and $\mathds{1}\big\lbrace{\normV[\big]{\mathbf{{w}}_{c,k}}_{2}^{2}} \big\rbrace \triangleq \normV[\big]{\normV[\big]{\mathbf{{w}}_{c,k}}_{2}^{2}}_0 $. This equivalence is of importance since the $l_0$-norm function can be approximated with a weighted $l_1$-norm convex function \cite{DaiY14}. To enable the use of such approximation in the context of our paper, we write the function $\normV[\big]{\normV[\big]{\mathbf{{w}}_{c, b,k}}_{2}^{2}}_0$ as a reweighed $l_1$-norm as follows:
\begin{equation}
\normV[\big]{\normV[\big]{\mathbf{{w}}_{c, b,k}}_{2}^{2}}_0 = \beta_{c,b,k} \normV[\big]{\mathbf{{w}}_{c, b,k}}_{2}^{2},
\end{equation}
\noindent where $\beta_{c,b,k}$ is a constant weight associated with BS $b$ in $\mathcal{B}_c$ and user $k$ and is defined in this work as
\begin{equation}\label{betaE}
\beta_{c,b,k} = \frac{1}{\delta + \normV[\big]{\mathbf{{w}}_{c, b,k}}_{2}^{2}},
\end{equation}
where $\delta > 0$ is a regularization constant\footnote{The regularization parameter can be chosen very small to make the approximation error arbitrary small. In the simulations we choose $\delta= 10^{-12}$ and we set the beamforming vector $\mathbf{w}_{c,b,k} = \mathbf{0}$ in iteration $t$ if $\normV[\big]{\mathbf{w}_{c,b, k}}_{2}^{2} \leq \delta$. This results in negligible error on the achievable data rate of user $k$.}. In a similar manner we define $\tilde{\beta}_{c,b,k}$ and $\beta_{c,k}$. When BS $b$ assigns low transmit power to user $k$, $\beta_{c,b,k}$ increases and thus user $k$ adds a non-desired burden to the fronthaul link and the energy consumption. The algorithm would then likely exclude $k$ from being served by $b$, thereby ensuring that the network only activates links with reasonable transmit powers.
Since the $l_1$-norm is applied to a quadratic function of the beamforming vectors, the resulting approximation is a smooth continuous function which is easier to optimize as compared to a non-smooth $l_0$-norm.
The weights in \eqref{betaE} are chosen in such a way that BSs with a small transmit power allocated to user $k$ get higher weights $\beta_{c,b,k}$ and eventually drop out of the cluster of BSs sharing the message of user $k$. Only those BSs which have reasonable transmit power allocated to user $k$ participate in the transmission to user $k$. The reformulated objective now reads as

\begin{equation}
	f_{2,\text{EE}}(c) \triangleq \frac{\sum\nolimits_{k \in \mathcal{K}} {R}_{c, k}}{P_c^{\text{Tx}} + p_{2,\text{EE}}(c) + P_c^{\text{Pr}}}, \label{eq:Obj4_}
\end{equation}
where
\begin{equation}
	p_{2,\text{EE}}(c) = \underbrace{\sum\limits_{(b, k) \in \mathcal{B}_c \times \mathcal{K}} \beta_{c,b,k}\normV[\big]{\mathbf{{w}}_{c, b,k}}_{2}^{2} P_{b}^{\text{fthl}}}_{\text{Fronthaul processing cost}} + \underbrace{\sum\limits_{(b, k) \in \mathcal{B}_c \times \mathcal{K}} \tilde{\beta}_{c,b,k}\normV[\big]{\mathbf{\tilde{w}}_{c, b,k}}_{2}^{2} P_{k}^{\text{proc}}}_{\text{Processing at BS}} + \underbrace{\sum\limits_{k \in \mathcal{K}} \beta_{c,k}\normV[\big]{\mathbf{w}_{c,k}}_{2}^{2} P_{k}^{\text{proc}}}_{\text{Processing at CP}}.
\end{equation}
Note that the function $p_{2,\text{EE}}(c)$ is a $l_0$-norm relaxed formulation of \eqref{hee}.
The reformulated problem \eqref{eq:Opt3} can now be written as
\begin{subequations}\label{eq:Opt4}
	\begin{align}
		&\underset{\mathbf{w}, \mathbf{\tilde{w}}, \bm {\gamma},\mathbf{r}}{\text{maximize}}\quad \sum_{c \in \mathcal{C}} f_{2,\text{EE}}(c) \label{eq:Obj4} \\
		&\text{subject to} \quad \eqref{eq:BM}, \eqref{eq:powe}, \eqref{eq:FrontC}, \eqref{eq:BgM}, \eqref{eq:BgM2}, \eqref{eq:q1}, \eqref{eq:q2}, \nonumber\\
		&\beta_{c,b,k}\normV[\big]{\mathbf{w}_{c,b, k}}_{2}^{2}  + \tilde{\beta}_{c,b,k}\normV[\big]{\mathbf{\tilde{w}}_{c, b, k}}_{2}^{2} \leq 1\label{eq:BME4} &&\forall k \in \mathcal{K}, \forall b \in \mathcal{B}_c, \forall c \in \mathcal{C},\\
		&C_{2,b}\left(\mathbf{w}\right) \leq F_{b,c} &&\forall b \in \mathcal{B}_c, \forall c \in \mathcal{C} \label{eq:FrontC4}.
	\end{align}
\end{subequations}
Note that constraints \eqref{eq:BME} and \eqref{eq:FrontC} are also replaced in \eqref{eq:Opt4} by \eqref{eq:BME4} and  \eqref{eq:FrontC4}, respectively.
Before reformulating the fronthaul capacity constraint \eqref{eq:FrontC}, we first elaborate on the second term in \eqref{eq:Fronthl}, namely $ (1-c_{f_k,b})\mathds{1}\left\lbrace{\normV[\big]{\mathbf{\tilde{w}}_{c, b, k}}_{2}^{2}}\right\rbrace $.
If a BS $b$ caches file $f_k$ then $c_{f_k,b} = 1$, which means the fronthaul link of BS $b$ is not affected by user $k$. Otherwise, if the BS $b$ does not cache the requested file by user $k$, i.e., $k$ is a cache-miss user from BS $b$ perspective, we know from previous definitions and \eqref{eq:BM} that $\tilde{\mathbf{w}}_{c,b,k}=\mathbf{0}_L$. Based on these observations, we can conclude that the second term in \eqref{eq:Fronthl} does not influence the fronthaul link and can thus be ignored. The reformulated fronthaul term is now
\begin{equation}
	C_{2,b}\left(\mathbf{w}\right) = \sum_{k \in \mathcal{K}}\beta_{c,b,k}\normV[\big]{\mathbf{w}_{c,b, k}}_{2}^{2} R_{c,k}.
\end{equation}
The above reformulations help overcoming the discrete nature of the original problem \eqref{eq:Opt3}. However, the reformulated problem \eqref{eq:Opt4} remains difficult, non-convex, and so it is tackled next using fractional programming and SICA.
\subsection{Fractional Programming and Successive Inner-Convex Approximations}
Note that the objective function in \eqref{eq:Obj4} is a sum of ratios of linear and convex functions, which makes \eqref{eq:Opt4} a suitable platform to apply a Dinkelbach-like algorithm \cite{8187084}\footnote{For a more general overview regarding fractional programming for EE maximization, especially a general formulation of Dinkelbach's algorithm, please refer to the work \cite{8187084}. A more detailed description of sequential optimization can also be found in \cite{7862919} and \cite{10.2307/169728}.}. Observe, however, that the non-convex feasible set of problem \eqref{eq:Opt4}, stemming from constraints \eqref{eq:FrontC4}, \eqref{eq:q1} and \eqref{eq:q2}, would hinder a direct application of a Dinkelbach-like algorithm to solve \eqref{eq:Opt4}, as this would require solving a non-convex problem to obtain a stationary solution, which is computationally prohibitive, especially when the network becomes reasonably sized \cite{7862919}. To overcome such difficulty, this paper uses an SICA approach so as to enable an efficient implementation of a Dinkelbach-like algorithm. 
A highlight of our proposed algorithm is that it can be implemented in a distributed fashion across the multiple CPs. 
We start by reformulating problem \eqref{eq:Opt4} to get a formulation that is amenable to apply SICA techniques.
\subsection{Convexification of Problem \eqref{eq:Opt4}}
First, we tackle \eqref{eq:FrontC4} by introducing ´the slack variables: $\mathbf{t}  = {\rm vec}(\left\lbrace {t}_{k, b} |\forall \left(k, b\right) \in  \mathcal{K}\times \mathcal{B}\right\rbrace)$, $\tilde{\mathbf{t}}  = {\rm vec}(\left\lbrace \tilde{t}_{k, b} | \forall \left(k, b\right) \in  \mathcal{K}\times \mathcal{B}\right\rbrace)$, and  $\mathbf{u}  = {\rm vec}(\left\lbrace {u}_{c, k} | \forall \left(c, k\right) \in \mathcal{C}\times\mathcal{K}\right\rbrace)$. 
Then, for all $k$, $b$, and $c$, define the following auxiliary constraints
\begin{align}
	&\beta_{c,b,k}\normV[\big]{\mathbf{w}_{c,b, k}}_{2}^{2} \leq t_{k,b} \label{eq:beta1},\\
	&\tilde{\beta}_{c,b,k}\normV[\big]{\mathbf{\tilde{w}}_{c, b, k}}_{2}^{2} \leq \tilde{t}_{k,b} \label{eq:beta2},\\
	&\sum_{b\in\mathcal{B}_c} \beta_{c,b,k}\normV[\big]{\mathbf{w}_{c,b, k}}_{2}^{2} \leq {u}_{c,k}. \label{eq:beta3}
\end{align}
The fronthaul capacity constraint \eqref{eq:FrontC4} is reformulated using the slack variable $\mathbf{t}$ as follows
\begin{equation}
	\sum_{k \in \mathcal{K}}t_{k,b} R_{c,k} \leq F_{b,c} \qquad \forall b \in \mathcal{B}_c, \forall c \in \mathcal{C}. \label{eq:FrontRe1}
\end{equation}
This function is non-convex as it is bilinear in the optimization variables. However, using some algebraic transformations of $\sum_{k \in \mathcal{K}}t_{k,b} R_{c,k}$, the above constraint can be equivalently written as

\begin{equation}
	\sum_{k\in\mathcal{K}}\frac{1}{4} \big(\underbrace{\left(t_{k,b}+R_{c,k}\right)^2}_{\text{convex}} - \underbrace{\left(t_{k,b} - R_{c,k}\right)^2}_{\text{convex}}\big) \leq F_{b,c}. \label{DC}
\end{equation}
The left hand side of \eqref{DC} is a difference of convex functions, which allows for applying SICA methods. The idea here is to find a convex surrogate upper-bound to the non-convex function associated with \eqref{DC}, i.e., $\sum_{k\in\mathcal{K}} t_{k,b} R_{c,k} - F_{b,c}$. This can be done by keeping the convex part and linearizing the concave one using first-order Taylor expansion. Define $g_1(\mathbf{t},\mathbf{r},\mathbf{t}',\mathbf{r}')$ as
\begin{equation}
	g_1(\mathbf{t},\mathbf{r},\mathbf{t}',\mathbf{r}') \triangleq \sum_{k \in \mathcal{K}}\left( \left(t_{k,b}+R_{c,k}\right)^2 - 2 \left(t'_{k,b}-R'_{c,k}\right)\left(t_{k,b}-R_{c,k}\right) + \left(t'_{k,b}-R'_{c,k}\right)^2\right) - 4 F_{b,c} \label{eq:FrontRe2}.
\end{equation}
Here $\mathbf{t}'  = {\rm vec}(\left\lbrace {t}'_{k, b} | \forall \left(k, b\right) \in  \mathcal{K}\times \mathcal{B}\right\rbrace)$ and $\mathbf{r}'  = {\rm vec}(\left\lbrace {R}'_{c, k} | \forall \left(c, k\right) \in  \mathcal{C}\times \mathcal{K}\right\rbrace)$ are feasible fixed values, which satisfy the pre-defined constraints \eqref{eq:beta1} and \eqref{eq:FrontRe1}. Such feasible fixed values would eventually be updated iteratively, so as to refine the feasible set at every iteration of the SICA.%
\begin{lemma} \label{lma1}
	For all feasible values $\left(\mathbf{t}',\mathbf{r}'\right)$ and all $(c,b)\in(\mathcal{C},\mathcal{B}_c)$, the function $g_1(\mathbf{t},\mathbf{r},\mathbf{t}',\mathbf{r}')$ satisfies
	\begin{equation} \label{eq:lma11}
		g_1(\mathbf{t},\mathbf{r},\mathbf{t}',\mathbf{r}') \geq \sum_{k\in\mathcal{K}} t_{k,b} R_{c,k} - F_{b,c}.
	\end{equation}
\end{lemma}%
\begin{proof}
	Please refer to Appendix \ref{app1}.
\end{proof}%
The above step allows us to convexify the fronthaul capacity constraint \eqref{eq:FrontC4}, by means of finding the convex surrogate upper-bound function $g_1(\mathbf{t},\mathbf{r},\mathbf{t}',\mathbf{r}')$. We next apply a similar procedure to the non-convex constraint \eqref{eq:q1}. To this end, based on \eqref{eq:q1}, we define ${\hat{g}}_2({\bm \gamma},\mathbf{r})$ as
\begin{equation}
	{\hat{g}}_2({\bm \gamma},\mathbf{r}) =  R_{c,k} - W\log_2\left(1+ \gamma_{c, k}\right) \leq 0. \label{g2}
\end{equation}
The function ${\hat{g}}_2({\bm \gamma},\mathbf{r})$ in \eqref{g2} is non-convex in $\gamma_{c, k}$. One can, however, linearize the non-convex part of ${\hat{g}}_2({\bm \gamma},\mathbf{r})$, namely $\log_2\left(1+ \gamma_{c, k}\right)$, around $\bm {\gamma}$ using the first-order Taylor expansion. The convex upper-bound of ${\hat{g}}_2({\bm \gamma},\mathbf{r})$, denoted by $g_2({\bm \gamma},\mathbf{r},{\bm \gamma}')$ can be then written as
\begin{equation}
	g_2({\bm \gamma},\mathbf{r},{\bm \gamma}')\triangleq \frac{R_{c,k}}{W} - \log_2\left(1+ \gamma'_{c, k}\right) - \frac{1}{\ln(2)\left(1+ \gamma'_{c, k}\right)} \left(\gamma_{c, k} - \gamma'_{c, k}\right), \label{eq:q1Re}
\end{equation}
\noindent where the variables ${\bm \gamma}'  = {\rm vec}(\left\lbrace {\gamma}'_{c, k} | \forall \left(c, k\right) \in  \mathcal{C}\times \mathcal{K}\right\rbrace)$ are feasible fixed values, which allows to convexify \eqref{eq:q1}.

Consider now constraint \eqref{eq:q2}, which can be re-written as
\begin{equation}
	\zeta^+(\mathbf{{w}},\mathbf{\tilde{w}}) - \zeta^-(\mathbf{{w}},\mathbf{\tilde{w}},{\bm \gamma}) \leq 0, \label{eq:h}
\end{equation}
where
\begin{equation}
	\zeta^+(\mathbf{{w}},\mathbf{\tilde{w}}) = \sigma^2  + \sum_{(c', k')\neq (c, k)}{\big|{\mathbf{h}}_{c', k}^\dagger\left( {\mathbf{w}_{c', k'}}+ {\mathbf{\tilde{w}}_{c', k'}}\right) }\big|^2 ,
\end{equation}
and
\begin{equation}
	\zeta^-(\mathbf{{w}},\mathbf{\tilde{w}},{\bm \gamma}) = \frac{\big|{\mathbf{h}}_{c, k}^\dagger\left( {\mathbf{w}_{c,k}}+{\mathbf{\tilde{w}}_{c,k}}\right) \big|^2}{\gamma_{c, k}}.
\end{equation}
The formulation \eqref{eq:h} is in DC form, since $\zeta^+$ is a convex, quadratic function in $\mathbf{{w}}$ and $\mathbf{\tilde{w}}$. $\zeta^-$ is also convex since it is a rational function with quadratic numerator and positive linear denominator \cite{convexOpt}. Lemma \ref{lma2} below states a viable first-order approximation of such function $\zeta^-$.
\begin{lemma} \label{lma2}
	Define $\zeta(\mathbf{x},\xi)$, where $\mathbf{x}\in\mathbb{C}^\text{P}$ and $\xi > 0$, as $\zeta(\mathbf{x},\xi)=\frac{|\mathbf{x}|^2}{\xi}$, then the first-order approximation of $\zeta(\mathbf{x},\xi)$ around a feasible point $(\mathbf{x}',\xi')$ satisfies
	\begin{equation}
		\zeta(\mathbf{x},\xi) \geq \tilde{\zeta}(\mathbf{x},\xi,\mathbf{x}',\xi') = \frac{2\Re\big\{ \left(\mathbf{x}'\right)^\dagger \mathbf{x} \big\}}{\xi'} - \frac{\xi}{\left(\xi'\right)^2}\big|\mathbf{x}'\big|^2. \label{eq:lma1}
	\end{equation}
\end{lemma}
\begin{proof}
	Please refer to Appendix \ref{app2}.
\end{proof}%
In order to obtain a convex reformulation of \eqref{eq:h}, we use the first order-approximaiton of $\zeta^-$ around the feasible point $(\mathbf{{w}}',\mathbf{\tilde{w}}',\bm {\gamma}')$ according to Lemma~\ref{lma2} to get:
\begin{align}
	\frac{\big|{\mathbf{h}}_{c, k}^\dagger\left( {\mathbf{w}_{c,k}}+{\mathbf{\tilde{w}}_{c,k}}\right) \big|^2}{\gamma_{c, k}} \approx \frac{2}{\gamma'_{c, k}}\Re\left\{ \left( {\mathbf{w}'_{c,k}}+{\mathbf{\tilde{w}}'_{c,k}}\right)^\dagger {\mathbf{h}}_{c, k} {\mathbf{h}}_{c, k}^\dagger \left( {\mathbf{w}_{c,k}}+{\mathbf{\tilde{w}}_{c,k}}\right) \right\} \nonumber\\
	- \frac{\gamma_{c, k}}{\left(\gamma'_{c, k}\right)^2} \big|{\mathbf{h}}_{c, k}^\dagger\left( {\mathbf{w}'_{c,k}}+{\mathbf{\tilde{w}}'_{c,k}}\right) \big|^2. \label{eq:q2Re} 
\end{align}
By substituting the linearized form above into the non-convex formulation \eqref{eq:h}, the inner convex approximation $g_3(\mathbf{w},\tilde{\mathbf{w}},{\bm \gamma},\mathbf{w}',\tilde{\mathbf{w}}',{\bm \gamma}')$ of $\zeta^+(\mathbf{{w}},\mathbf{\tilde{w}}) - \zeta^-(\mathbf{{w}},\mathbf{\tilde{w}},{\bm \gamma})$ can be written as:
\begin{align}
	g_3(&\mathbf{w},\tilde{\mathbf{w}},{\bm \gamma},\mathbf{w}',\tilde{\mathbf{w}}',{\bm \gamma}')\triangleq\sigma^2  + \sum_{(c', k')\neq (c, k)}{\big|{\mathbf{h}}_{c', k}^\dagger\left( {\mathbf{w}_{c', k'}}+ {\mathbf{\tilde{w}}_{c', k'}}\right) }\big|^2  \label{eq:g3_}\\
	&- \sum_{c\in\mathcal{C}}\left[\frac{2}{\gamma'_{c, k}}\Re\left\{ \left( {\mathbf{w}'_{c,k}}+{\mathbf{\tilde{w}}'_{c,k}}\right)^\dagger {\mathbf{h}}_{c, k} {\mathbf{h}}_{c, k}^\dagger \left( {\mathbf{w}_{c,k}}+{\mathbf{\tilde{w}}_{c,k}}\right) \right\}
	+ \frac{\gamma_{c, k}}{\left(\gamma'_{c, k}\right)^2} \big|{\mathbf{h}}_{c, k}^\dagger\left( {\mathbf{w}'_{c,k}}+{\mathbf{\tilde{w}}'_{c,k}}\right) \big|^2\right].\nonumber
\end{align}
Note that $\mathbf{{w}}'  = {\rm vec}(\left\lbrace {w}'_{c, k} | \forall \left(c, k\right) \in  \mathcal{C}\times \mathcal{K}\right\rbrace)$, $\mathbf{\tilde{w}}'  = {\rm vec}(\left\lbrace \tilde{w}'_{c, k} | \forall \left(c, k\right) \in  \mathcal{C}\times \mathcal{K}\right\rbrace)$ and the ${\bm \gamma}'$ are feasible fixed values satisfying constraints \eqref{eq:powe}, \eqref{eq:BgM}, \eqref{eq:q2}, \eqref{eq:BME4}, \eqref{eq:FrontC4}, and \eqref{g2}.\\
As a last reformulation step, we rewrite the EE function \eqref{eq:Obj4_} as follows:
\begin{equation}
	f_{3,\text{EE}}(c) \triangleq \frac{\sum\nolimits_{k \in \mathcal{K}} {R}_{c, k}}{\sum\limits_{(b, k) \in \mathcal{B}_c \times \mathcal{K}} t_{k,b} P_{b}^{\text{fthl}} + \sum\limits_{(b, k) \in \mathcal{B}_c \times \mathcal{K}} \tilde{t}_{k,b} P_{k}^{\text{proc}} + \sum\limits_{k \in \mathcal{K}} {u}_{c,k} P_{k}^{\text{proc}}+ P_c^{\text{Tx}}+P_c^{\text{Pr}}}. \label{eq:f3}
\end{equation}
As shown implicitly from equation \eqref{eq:f3}, the slack variables $t_{k,b}$, $\tilde{t}_{k,b}$, and ${u}_{c,k}$ denote at which point of the system the data from user $k$ is processed and distributed. To be more precise, the variables ${u}_{c, k}$ and ${t}_{k, b}$ refer to processing the data of user $k$ at cloud $c$. The variable ${u}_{c, k}$ refers to the data from user $k$ being processed at cloud $c$ and then being forwarded to all BSs participating in serving $k$, i.e., ${t}_{k, b}$ at all BSs $b\in\mathcal{B}_c$. On the other hand, $\tilde{t}_{k, b}$ refers to processing the data of user $k$ locally at BS $b$ in case the requested file for user $k$ is cached there, i.e., $c_{f_k,b}=1$.\\
Combining the above confexification steps of all constraints of the original problem \eqref{eq:Opt4} gives the following reformulated problem:
\begin{subequations}\label{eq:Opt5}
	\begin{align}
		&\underset{\mathbf{y}}{\text{maximize}}\quad \sum_{c \in \mathcal{C}} f_{3,\text{EE}}(c) \label{eq:Obj5} \\
		&\text{subject to} \quad \eqref{eq:BM}, \eqref{eq:powe}, \eqref{eq:BgM}, \eqref{eq:BgM2}, \eqref{eq:BME4}, \eqref{eq:beta1}, \eqref{eq:beta2}, \eqref{eq:beta3}, \nonumber\\
		&g_1(\mathbf{t},\mathbf{R};\mathbf{t}',\mathbf{R}') \leq 0 &&\forall b\in\mathcal{B}_c, \forall c\in\mathcal{C}, \label{eq:g1}\\
		&g_2({\bm \gamma},\mathbf{R};{\bm \gamma}') \leq 0 &&\forall k\in\mathcal{K}, \forall c\in\mathcal{C},\label{eq:g2}\\
		&g_3(\mathbf{w},\tilde{\mathbf{w}},{\bm \gamma};\mathbf{w}',\tilde{\mathbf{w}}',{\bm \gamma}') \leq 0 &&\forall k\in\mathcal{K}, \forall c\in\mathcal{C}.\label{eq:g3}
			\end{align}
\end{subequations}
where the optimization is over $\mathbf{y}$ defined as $\mathbf{y}=\left[ \mathbf{w}^T, \mathbf{\tilde{w}}^T, \mathbf{t}^T, \tilde{\mathbf{t}}^T, \mathbf{u}^T, \bm {\gamma}^T, \mathbf{r}^T \right]^T$ which contains all optimization variables, and where $\mathbf{y}'=\left[ \mathbf{w}'^T, \mathbf{\tilde{w}}'^T, \mathbf{t}'^T, \tilde{\mathbf{t}}'^T, \mathbf{u}'^T, \bm {\gamma}'^T, \mathbf{r}'^T \right]^T\in\mathcal{Y}$ is a vector containing all fixed values, where $\mathcal{Y}$ is the convex feasible set defined by the constraints \eqref{eq:BM}, \eqref{eq:powe}, \eqref{eq:BgM}-\eqref{eq:BgM2}, \eqref{eq:BME4}, \eqref{eq:beta1}-\eqref{eq:beta3}, and \eqref{eq:g1}-\eqref{eq:g3}.
\subsection{Iterative Algorithm}
Despite the non-convexity of the fractional function $f_{3,\text{EE}}(c)$ in \eqref{eq:f3}, all constraints of problem \eqref{eq:Opt5} are convex, and so \eqref{eq:Opt5} can be iteratively solved using a SICA and Dinkelbach-like algorithm. More precisely, in order to apply a Dinkelbach-like algorithm, we define $g_{4}(c)$ and $g_{5}(c)$ as the numerator and denominator of $f_{3,\text{EE}}(c)$, respectively:
\begin{equation}
	g_{4}(c) \triangleq \sum\nolimits_{k \in \mathcal{K}} {R}_{c, k}, \label{eq:fn}
\end{equation}
and
\begin{equation}
	g_{5}(c) \triangleq \sum\limits_{(b, k) \in \mathcal{B}_c \times \mathcal{K}} t_{k,b} P_{b}^{\text{fthl}} + \sum\limits_{(b, k) \in \mathcal{B}_c \times \mathcal{K}} \tilde{t}_{k,b} P_{k}^{\text{proc}} + \sum\limits_{k \in \mathcal{K}} {u}_{c,k} P_{k}^{\text{proc}}+ P_c^{\text{Tx}}+P_c^{\text{Pr}}. \label{eq:fd}
\end{equation}
We then iteratively search for a unique solution to the following auxiliary convex optimization problem:
\begin{equation}\label{eq:Fcl}
	F({\bm \lambda_{j}}) = \underset{\mathbf{y}\in\mathcal{Y}}{\text{max}} \left\{ \sum_{c\in\mathcal{C}}g_{4}(c) - \lambda_{j}(c) g_{5}(c) \right\},
\end{equation}
where ${\bm \lambda_{j}} = {\rm vec}(\{\lambda_{j}(c)|\,\forall c\in\mathcal{C}\})$ is a constant vector that is updated after each iteration as follows:
\begin{equation}\label{eq:lambda}
	\lambda_{j+1}(c) = \frac{g_{4}(c)}{g_{5}(c)}, \quad\forall c\in\mathcal{C}.
\end{equation}
To solve problem \eqref{eq:Opt5}, we distinguish between an outer and an inner loop. In the outer loop, we update the feasible fixed values for SICA, initialize ${\bm\lambda_{0}}$ so it would be used in the inner loop, and check for convergence. In the inner loop, we use the Dinkelbach-like algorithm, and solve $F({\bm \lambda_{j}})$ iteratively using \eqref{eq:Fcl}-\eqref{eq:lambda}. Such approach produces a solution to the underlying fractional program \eqref{eq:Opt5} with optimal values $\mathbf{\hat{y}}_\nu$. At iteration $\nu$ of the outer loop, we refine the feasible set $\mathbf{y}'$ using the optimal values $\mathbf{\hat{y}}_\nu$ as fixed values for the next iteration. The algorithm stops at convergence. The steps of such approach are presented in Algorithm \ref{alg:sca_dinkel}. 
\begin{algorithm}
	\caption{Combined SICA and Dinkelbach-like algorithm.} \label{alg:sca_dinkel}
	\begin{algorithmic}[1]
		\STATE $\nu = 0; \;\mathbf{y}'\in \mathcal{Y};\; f_{3,\text{EE}}^{-1}(c) = 0, \forall c\in\mathcal{C};$
		\WHILE{$\left|\sum_{c\in\mathcal{C}} (f_{3,\text{EE}}^{\nu}(c) - f_{3,\text{EE}}^{\nu-1}(c)) \right| > \epsilon$}
		\STATE $j = 0,{\bm\lambda_{j}}$ with $ F({\bm\lambda_{j}}) \geq 0;$
		\WHILE{$ F({\bm\lambda_{j}}) > \epsilon $}
		\STATE $\mathbf{\hat{y}}_\nu = \text{arg}\,\underset{\mathbf{y}\in\mathcal{Y}}{\text{max}} \left\{ \sum_{c\in\mathcal{C}} g_{4}(c) - \lambda_{j}(c) g_{5}(c) \right\}; $
		\STATE $ \lambda_{j+1}(c) = \frac{g_{4}(c)}{g_{5}(c)},\forall c\in\mathcal{C};$
		\STATE $j = j + 1;$
		\ENDWHILE
		\STATE $\mathbf{y}' = \mathbf{\hat{y}}_\nu;$
		\STATE $\nu=\nu+1;$
		\ENDWHILE
	\end{algorithmic}%
\end{algorithm}%
To start the algorithm, the fixed values $\mathbf{y}'$ are computed. First, the beamforming vectors are initialized with feasible MRC beamformers \cite{8187586}. Please note that in order to determine $\mathbf{\tilde{w}}'$, the cache placement and the user requests have to be known. Based on these beamformers, the variables ${\bm \gamma}',\mathbf{r}'$ can be computed using equations \eqref{eq:q1} and \eqref{eq:q2}. At last, we compute $\mathbf{t}',\tilde{\mathbf{t}}',\mathbf{u}'$ by replacing the inequalities in \eqref{eq:beta1}-\eqref{eq:beta3} with equalities. To initialize ${\bm\lambda_{0}}$ with $F({\bm \lambda_{j}})\geq 0$, we use the feasible fixed values $\mathbf{y}'$ and compute
\begin{equation}
	\lambda_{0}(c) = \frac{g_{4}(c)}{g_{5}(c)}, \quad\forall c\in\mathcal{C}.
\end{equation}
Although the proposed solution does not guarantee the global optimality of the original complicated mixed-integer non-convex optimization problem \eqref{eq:Opt1}, our numerical simulations illustrate the appreciable performance improvement of the proposed algorithm as compared to state-of-the-art solutions. The numerical results further highlight the fast convergence of Algorithm~\ref{alg:sca_dinkel}, as illustrated in Sec.~\ref{ssec:Conv}.

\subsection{Distributed Implementation}
This part illustrates how Algorithm \ref{alg:sca_dinkel} can be implemented in a distributed manner across the multiple CPs. Firstly, the user-to-cloud association problem \eqref{eq:Opt2} can be done on a per-cloud basis using an iterative auction algorithm \cite{7086838}. Secondly, given the set of users $\mathcal{K}_c$ served by CP $c$, i.e., $\mathcal{K}_c \triangleq \{ k\in\mathcal{K} | \hspace{1mm} \exists k\in \mathcal{K}: z_{c,k} = 1 \}$, we define local beamforming vectors associated with cloud $c$ as $\mathbf{w}_c = {\rm vec}(\left\lbrace \mathbf{w}_{c , k}| \forall k \in \mathcal{K}_c \right\rbrace) $, $\mathbf{\tilde{w}}_c = {\rm vec}(\left\lbrace \mathbf{\tilde{w}}_{c , k}|k \in \mathcal{K}_c\right\rbrace) $, the serving clusters effectively reduced to $\mathbf{t}_c  = {\rm vec}(\left\lbrace {t}_{k, b} | \left(k, b\right) \in  \mathcal{K}_c\times \mathcal{B}_c\right\rbrace)$, $\tilde{\mathbf{t}}_c  = {\rm vec}(\left\lbrace \tilde{t}_{k, b} | \left(k, b\right) \in  \mathcal{K}_c\times \mathcal{B}_c\right\rbrace)$, $\mathbf{u}_c  = {\rm vec}(\left\lbrace {u}_{c, k} | k \in \mathcal{K}_c\right\rbrace)$, also $\bm \gamma_c = {\rm vec}(\left\lbrace \gamma_{c, k}| k \in  \mathcal{K}_c \right\rbrace) $ and $\mathbf{R}_c  = {\rm vec}(\left\lbrace {R}_{c, k} | k \in  \mathcal{K}_c\right\rbrace)$. Each cloud $c$ would then be able to solve problem \eqref{eq:Opt5} locally, via exchanging the interference terms $\sum_{(c',k')\neq(c,k)}\big|{\mathbf{h}}_{c', k}^\dagger\left( {\mathbf{w}_{c',k'}}+{\mathbf{\tilde{w}}_{c',k'}}\right) \big|^2$ with all other clouds $c'\neq c$, required for constraint \eqref{eq:g3}. In fact, as per Algorithm \ref{alg:sca_dinkel}, a distributed formulation necessitates the change of a few selected steps. More specific, in step $2$ of Algorithm \ref{alg:sca_dinkel}, the convergence criteria have to be checked per cloud individually. Also, steps $4$ to $8$, i.e., the inner loop, are performed per cloud leaving step $5$ as a local optimization problem at CP $c$. Therefore, the CPs would exchange interference information in every iteration of the outer loop as an additional step, i.e., between steps $8$ and $9$, which would enable the overall distributed implementation of the algorithm. \\
\indent While considering the distributed implementation of Algorithm \ref{alg:sca_dinkel} locally, i.e., the local step of the algorithm at one cloud only, the following theorem describes the guaranteed convergence to a stationary point.
\begin{theorem} \label{theorem}
	The distributed implementation of Algorithm \ref{alg:sca_dinkel}, while executed at cloud $c$, converges to a stationary point of the $l_0$-relaxed, distributed version of problem \eqref{eq:Opt4}, i.e., the relaxed problem at cloud $c$, given the assumption of fixed interference from all other clouds $c\neq c'$.
\end{theorem}
\begin{proof}
	Please refer to Appendix \ref{app3}.
\end{proof}

\subsection{Complexity Analysis}
Now we focus on the overall computational complexity of our proposed method. Starting with the inner loop that utilizes a Dinkelbach-like algorithm, the overall complexity depends on each subproblems' complexity as well as the convergence rate of the auxiliary problem series \eqref{eq:Fcl}-\eqref{eq:lambda}. Each subproblem \eqref{eq:Fcl} has a quadratic convex objective subject to quadratic convex constraints, and hence be cast as a second order cone program (SOCP) \cite{Lobo1998ApplicationsOS}. Such problems can be solved using interior-point methods. Since the total number of variables for each subproblem is given by $d_1=(\text{K}(2\text{B}(\text{L}+1)+3))$, the complexity metric becomes $\mathcal{O}((d_1)^{3.5})$. We let $\text{V}_{\text{1,max}}$ be the worst-case fixed number of iterations for convergence of the Dinkelbach-like algorithm. Since no optimization problem is solved in the outer loop, we can define $\text{V}_{2,\text{max}}$ as the worst-case fixed number of iterations needed for it to converge. The overall computational complexity of Algorithm \ref{alg:sca_dinkel} becomes, therefore, polynomial in the order of $\mathcal{O}(\text{V}_{\text{1,max}}\text{V}_{\text{2,max}}(d_1)^{3.5})$, which is an upper bound on the complexity metric. We note, finally, that our proposed method consists of two instances of Algorithm \ref{alg:sca_dinkel}, one for determining the serving clusters and one for finding a high-quality solution of beamforming vectors. The second instance of Algorithm \ref{alg:sca_dinkel} operates on the sparse optimization problem \eqref{eq:Opt6} with even fever optimization variables, which typically requires fewer iterations; thereby reducing the overall complexity of the algorithm implementation.

\subsection{Fixed Clustering-Based Baseline}
Algorithm \ref{alg:sca_dinkel}, particularly, finds optimal association variables $\mathbf{t},\tilde{\mathbf{t}},\mathbf{u}$, which defines the user-to-BS association, as known as the clustering strategy. To benchmark our solution, we now fix the clustering strategy, and focus on finding optimal beamforming vectors by revisiting problem \eqref{eq:Opt5}. The optimization variables now become group sparse variables $\mathbf{y}_2=\left[ \mathbf{w}^T, \mathbf{\tilde{w}}^T, \bm {\gamma}^T, \mathbf{r}^T \right]^T$. The fixed feasible variables are $\mathbf{y}'_2=\left[ \mathbf{w}'^T, \mathbf{\tilde{w}}'^T, \bm {\gamma}'^T, \mathbf{r}'^T \right]^T$. The optimization problem with fixed clusters can be written as
\begin{subequations}\label{eq:Opt6}
	\begin{align}
		&\underset{\mathbf{y}_2}{\text{maximize}}\quad \sum_{c \in \mathcal{C}} f_{3,\text{EE}}(c) \label{eq:Obj6}, \\
		&\text{subject to} \quad \eqref{eq:powe},\eqref{eq:BgM},\eqref{eq:BgM2},\eqref{eq:FrontRe1} \nonumber\\
		&g_2({\bm \gamma},\mathbf{r};{\bm \gamma}',\mathbf{r}') \geq 0 &&\forall k\in\mathcal{K}, \forall c\in\mathcal{C},\\
		&g_3(\mathbf{w},\tilde{\mathbf{w}},{\bm \gamma};\mathbf{w}',\tilde{\mathbf{w}}',{\bm \gamma}') \leq 0 &&\forall k\in\mathcal{K}, \forall c\in\mathcal{C}.
	\end{align}
\end{subequations}
A simpler version of Algorithm \ref{alg:sca_dinkel} is used to solve \eqref{eq:Opt6}, where the set of optimization variables is reduced to $\left\{ \mathbf{w}, \mathbf{\tilde{w}}, \bm {\gamma}, \mathbf{r} \right\}$. Such simplified approach is used to assess the performance of our proposed algorithm, as shown next.

\section{Numerical Simulations}
In this section, we present numerical simulations that illustrate the performance of our proposed algorithms. Consider an MC-RAN scenario occupying a square area of [-400 400] $\times$ [-400 400] $\text{m}^2$. The BSs and the users are randomly placed in the studied MC-RAN. The distribution is uniform. 
Each BS is equipped with $L = 2$ transmit antennas and all BSs share the same fronthaul capacity constraint. The maximum transmit power is set to $32$ dBm for each BS. The channel model used for our simulations consists of a path-loss model $\text{PL}_{b,k} = 128.1 + 37.6\log_{10}(d_{b,k}),$ where $d_{b,k}$ is the distance between BS $b$ and user $k$ in km, a log-normal shadowing with 8dB standard deviation, and a Rayleigh channel fading with zero mean and unit variance. The noise power $\sigma^2$ is set to $ -102 + 10\log_{10}(W) + n_f $ dBm, where the channel bandwidth is set to $W = 10$ MHz and the noise figure to $n_f = 15$ dBm. The number of total files for caching is $F_b = 100$, we adopt the \emph{popularity aware} cache placement scheme from \cite{7925639}. The local memory size at each BS is considered to be $10$ files unless otherwise mentioned. As for the popularity of the files, we use the Zipf distribution \cite{7925639} with parameter $a = 0.15$. Hence, for the content request mode of each individual user, it is assumed that each user requests a random file following the given popularity distribution. We choose the convergence parameters of Algorithm \ref{alg:sca_dinkel} as $\epsilon_1 = 0.2$ and $\epsilon_2 = 0.03$. Regarding the costs of the EE metric, $P_b^{\text{fthl}}$ is chosen to be $40$\% of the processing power, $P_k^{\text{proc}} = 20 \;$dBm, and $P_c^{\text{Pr}} = 38\;$dBm  %
unless specified otherwise \cite{7437385}. At last, we define the number of users to be $28$ and the number of BSs to be $10$ unless mentioned otherwise. \\
We propose optimizing beamforming vectors and serving clusters jointly in a dynamic clustering scheme. To best benchmark our methods, we use a static clustering scheme as a baseline to our proposed algorithm, where predetermined fixed clusters are used instead and the optimization is carried out on the beamforming vectors only.
To determine such clusters, we use a load balancing algorithm applied in \cite{DaiY14} for the case of a single cloud. This benchmark is referred to as fixed clustering.
Both schemes can be implemented either in a centralized or distributed fashion. The former algorithm is implemented at one CP, processing data from all clouds, while the latter algorithm is implemented at every CP, managing their respective computations.
\subsection{Impact of Fronthaul Capacity} \label{Sim:FronthaulCapacity}
\begin{figure}
	\centering
	\begin{subfigure}[b]{0.49\textwidth}
		\includegraphics[width=\linewidth]{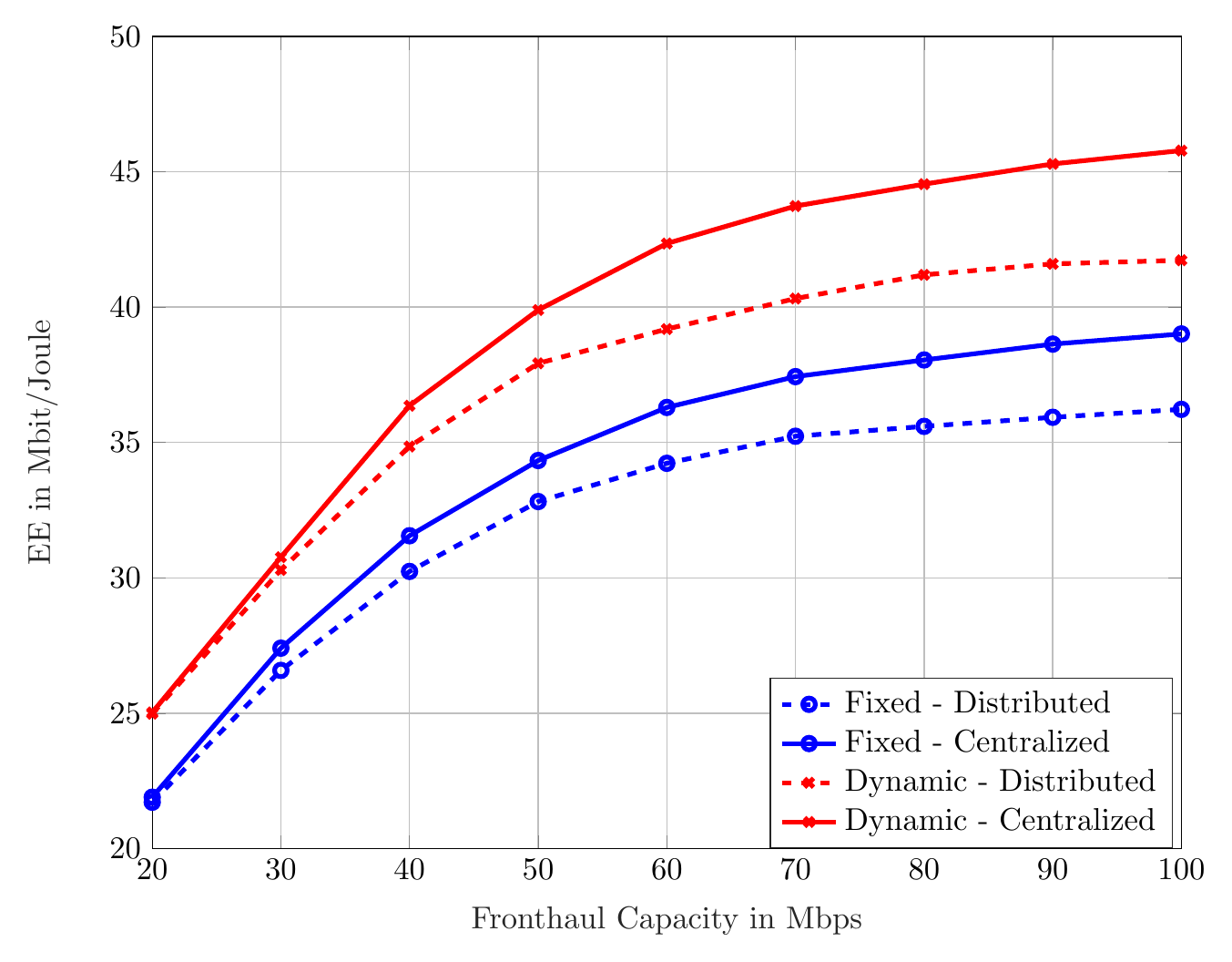}
		\caption{Cache holds up to $10$ files.}
		\label{yEE_xFthl_xcach10}
	\end{subfigure}\hfill
	\begin{subfigure}[b]{0.49\textwidth}
		\includegraphics[width=\linewidth]{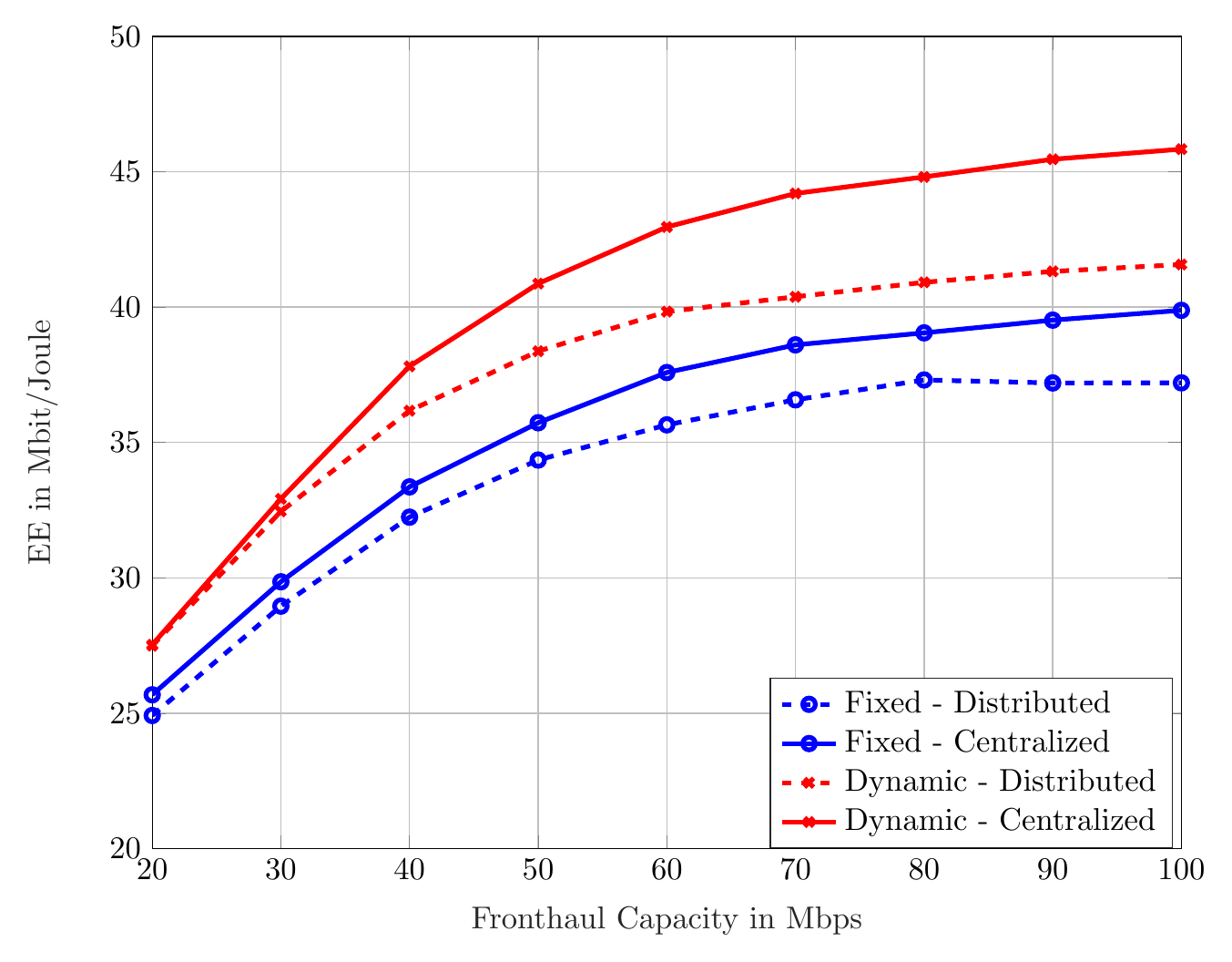}
		\caption{Cache holds up to $20$ files.}
		\label{yEE_xFthl_xcach20}
	\end{subfigure}
	\caption{EE as a function of fronthaul capacity for different cache sizes.} \label{yEE_xFthl_xcach}
	\vspace{-0.7cm}
\end{figure}
First, in Fig. \ref{yEE_xFthl_xcach}, we evaluate the performance of the two schemes, static and dynamic clustering, both using centralized and distributed implementations. Fig. \ref{yEE_xFthl_xcach} shows the EE as a function of the fronthaul capacity for two different cache sizes, i.e., $10$ in Fig. \ref{yEE_xFthl_xcach10}, and $20$ in Fig. \ref{yEE_xFthl_xcach20}. Both figures show that our proposed dynamic clustering outperforms the fixed clustering approaches regardless of being implemented distributively or centrally. Particularly, since BSs might drop out of serving clusters due to overloading or power constraints, the need for dynamic clustering emerges, as such situations cannot be compensated by a fixed cluster.\\ 
Further, the centralized implementation outperforms the distributed implementation for both schemes in Fig. \ref{yEE_xFthl_xcach10} and Fig. \ref{yEE_xFthl_xcach20}. Note that the gain of using a centralized instead of distributed implementation increases jointly with fronthaul capacity, i.e., the gap widens. 
In a low-fronthaul regime, the difference of both iterations is visibly insignificant, which highlights the role of our proposed distributed algorithm in limited fronthaul-capacity regimes, i.e., in cases where alleviating the fronthaul congestion is mostly required, as the performances of distributed and centralized algorithms become relatively similar.\\
A general observation from comparing Fig. \ref{yEE_xFthl_xcach10} and \ref{yEE_xFthl_xcach20} is that the EE gain increases when bigger cache sizes at the BSs are employed. For the dynamic centralized scheme, there is a $10$\% gain at $20\;$Mbps. This is particularly the case since the EE metric benefits from cache hits, mainly because a user can be served without utilizing the fronthaul link while requiring processing costs at the respective BS. At lower capacity regimes, fronthaul capacity is a scarce resource, and so the activation of a fronthaul link is a sensible decision, and the EE metric gains significantly from cache hits in such situations. 

\subsection{Processing Power vs. Caching Gain} \label{pprocvscach}
\begin{figure}
	\centering
	\begin{subfigure}[b]{0.49\textwidth}
		\includegraphics[width=\linewidth]{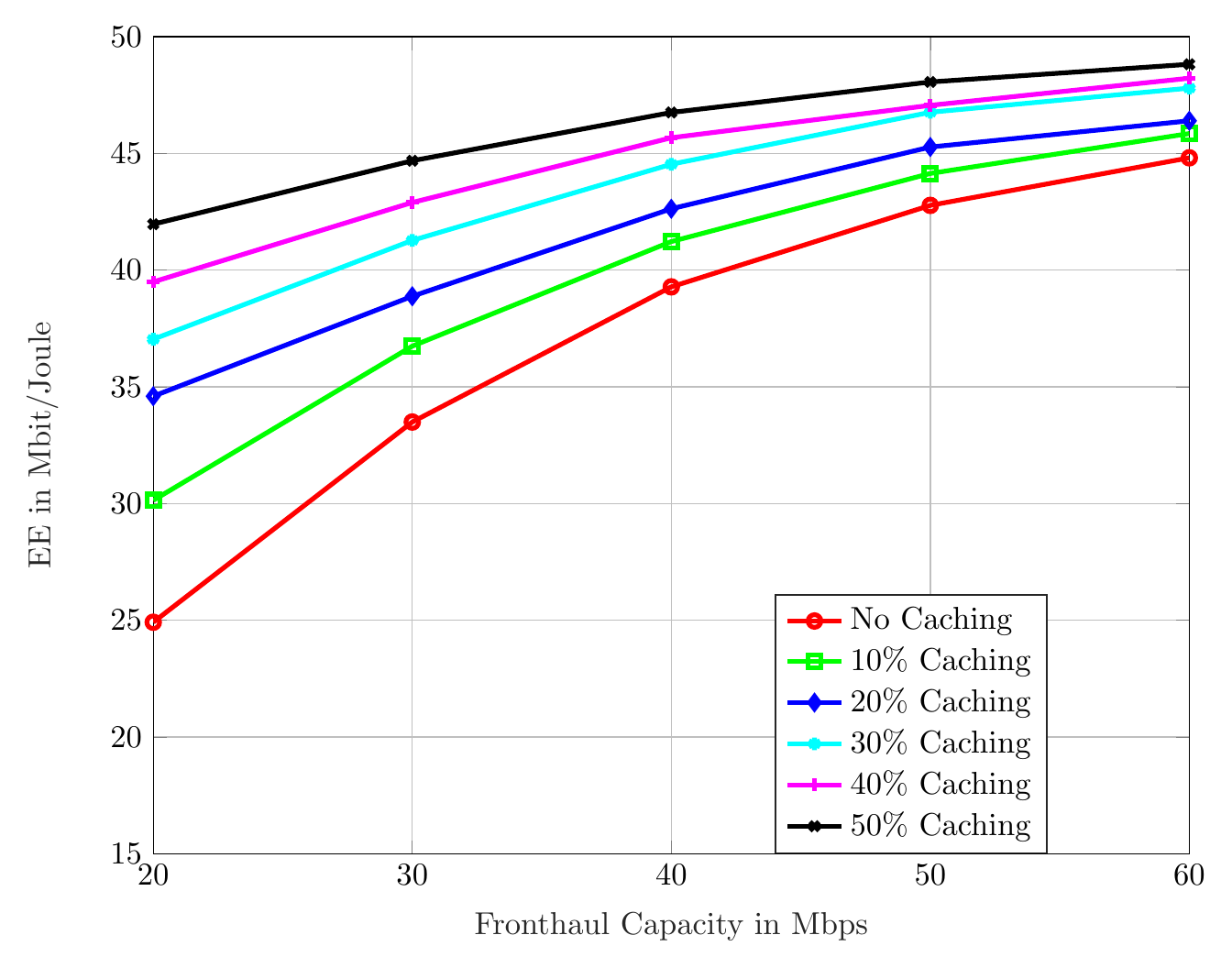}
		\caption{$P_k^{\text{proc}} = 10\;$dBm.}
		\label{yEE_xFthl_px10}
	\end{subfigure}\hfill
	\begin{subfigure}[b]{0.49\textwidth}
		\includegraphics[width=\linewidth]{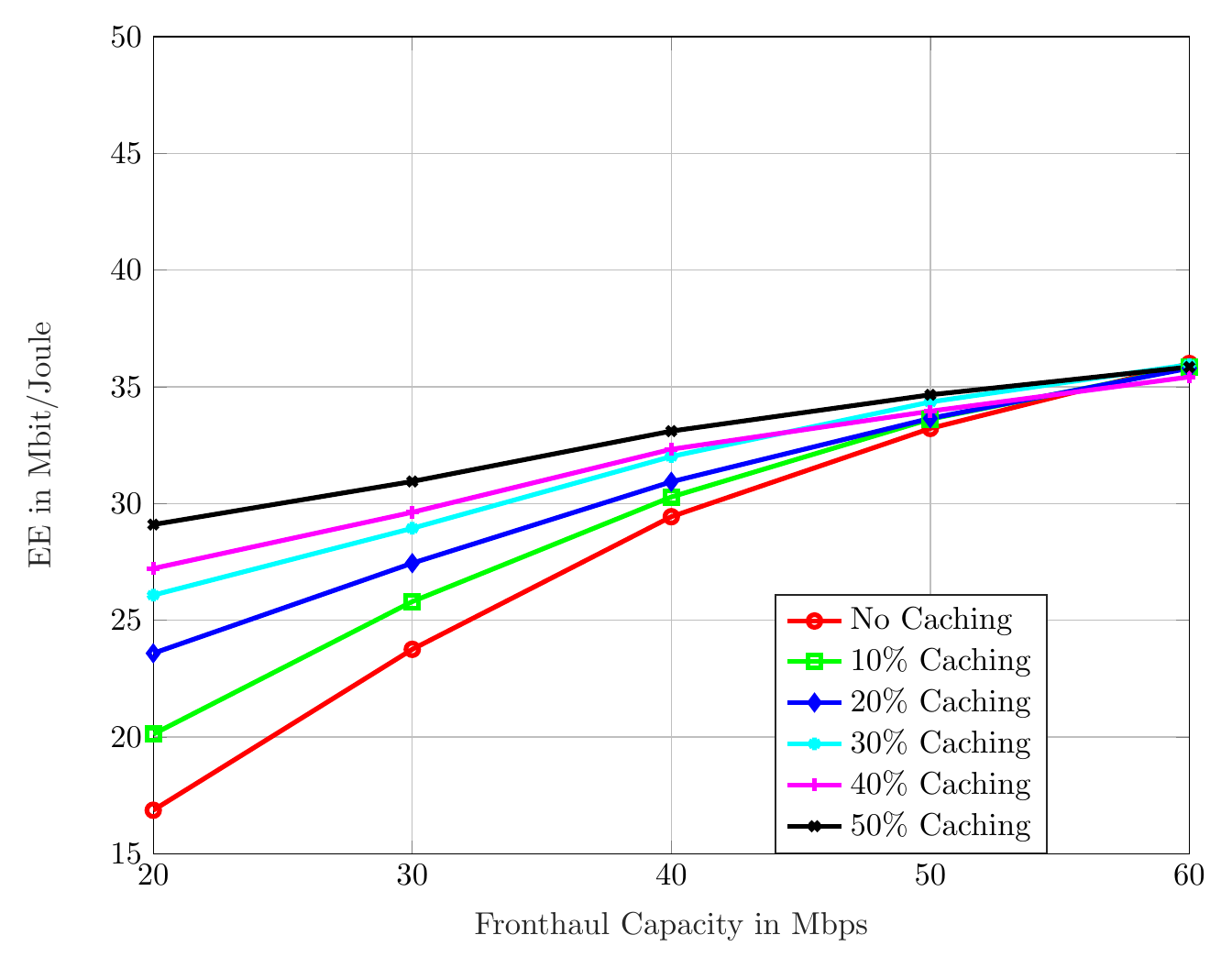}
		\caption{$P_k^{\text{proc}} = 25\;$dBm.}
		\label{yEE_xFthl_px25}
	\end{subfigure}
	\caption{EE as a function of fronthaul capacity, where different processing powers are utilized.}
	\label{yEE_xFthl_px}
	\vspace{-0.7cm}
\end{figure}
In the second set of simulations, we consider only the distributed implementation of Algorithm \ref{alg:sca_dinkel} with dynamic clustering. In Fig. \ref{yEE_xFthl_px}, we compare the EE for two different processing costs. We also vary the cache sizes for fixed costs, i.e., no cache up to a cache size of $50$ files.\\
Intuitively, a higher processing cost would reduce the EE, as this factor can not be compensated. Such behavior can be recognized comparing Fig. \ref{yEE_xFthl_px10} and Fig. \ref{yEE_xFthl_px25}. The EE for a cache size of $10$ files at $40\;$Mbps fronthaul capacity decreases by $26.54\;$\% as the processing cost increase from $10\;$dBm to $25\;$dBm.\\
In Fig. \ref{yPer_xFthl}, we plot the caching percentage gain, i.e., the EE gain of using a cache of $10$ files over not using cache, versus the fronthaul capacity for the two processing powers, so as to show the gain of the green line with respect to the red line of Fig. \ref{yEE_xFthl_px10} and Fig. \ref{yEE_xFthl_px25}. Fig. \ref{yPer_xFthl}, indeed, reaffirms that the EE gain from a bigger cache decreases with increasing fronthaul capacity. Interestingly, the gain of utilizing a bigger cache size also decreases with higher processing costs. In fact, Fig. \ref{yEE_xFthl_px25} and Fig.~\ref{yPer_xFthl} show that there is no conceivable gain from utilizing caching capabilities for $P_k^{\text{proc}} = 25\;$dBm at $60$ Mbps fronthaul capacity.
\begin{figure}
\centering
\begin{subfigure}[b]{0.49\textwidth}
	\includegraphics[width=\linewidth]{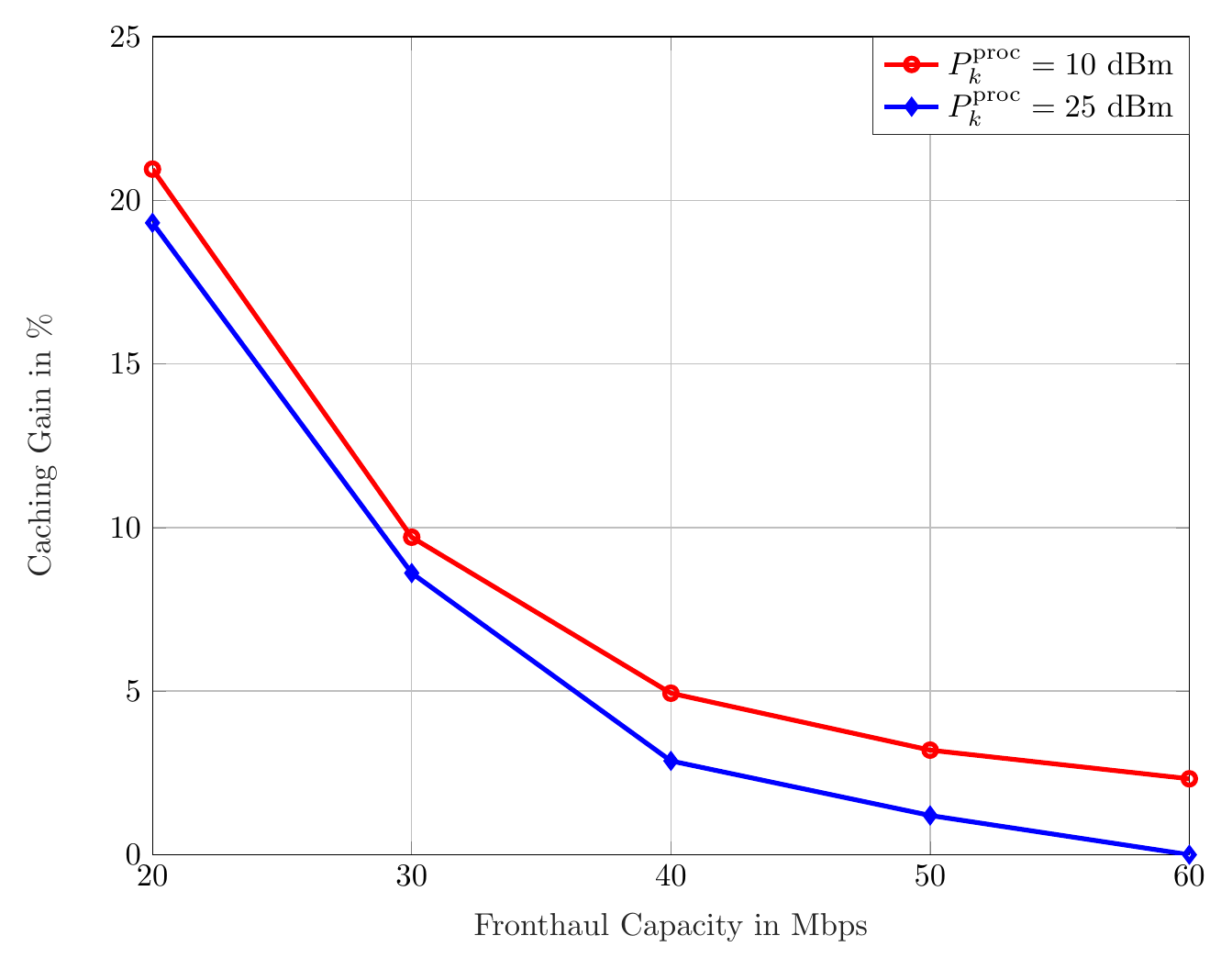}
	\caption{Caching gain in \% as a function of the fronthaul capacity for different processing powers.} \label{yPer_xFthl}
\end{subfigure}\hfill
\begin{subfigure}[b]{0.49\textwidth}
	\includegraphics[width=\linewidth]{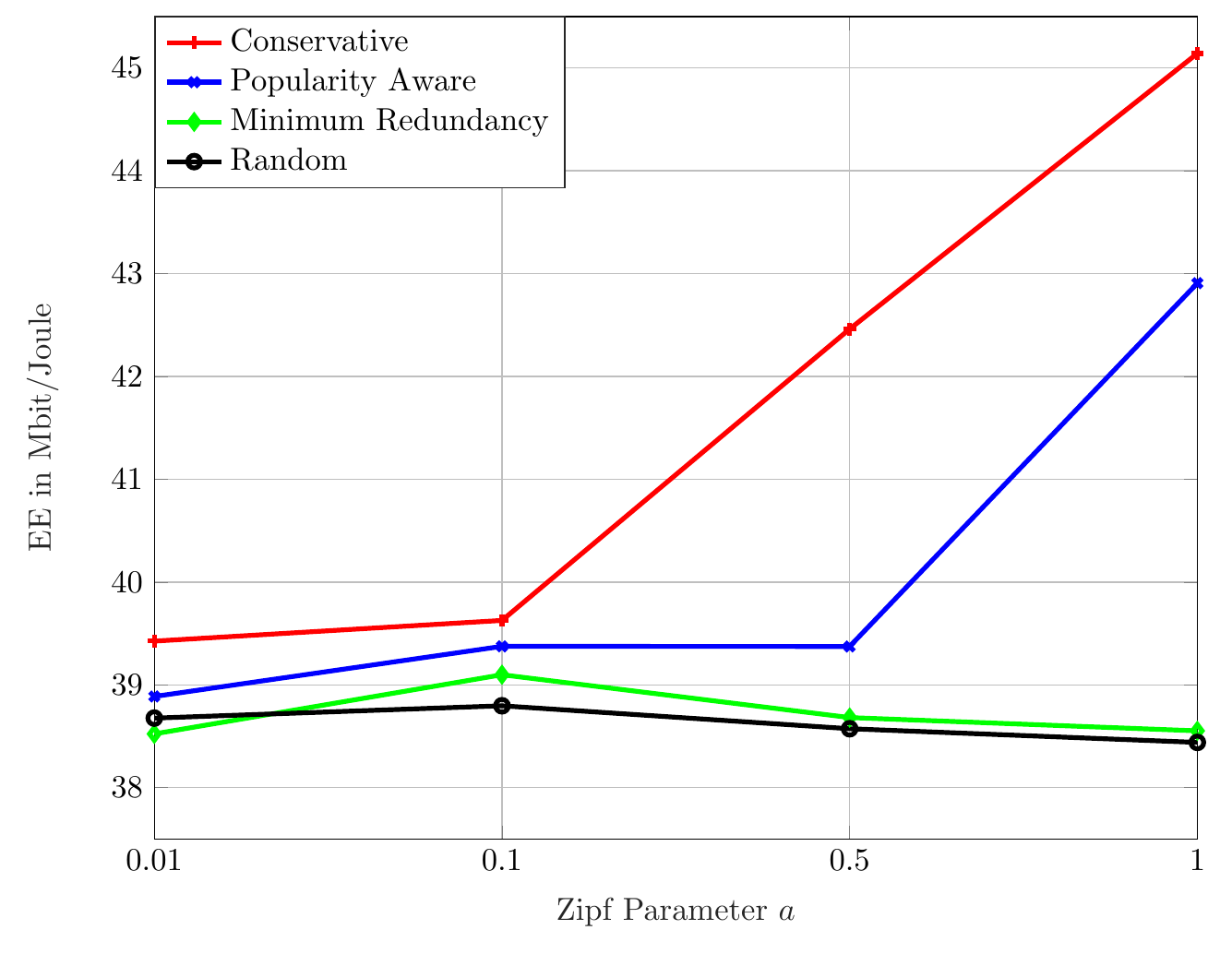}
	\caption{EE as a function of the Zipf paramter $a$, a comparison of cache placement strategies.} \label{yEE_xAlph_10cach}
\end{subfigure}
\caption{Caching gain and EE as functions of different system parameters.}
\label{yEEPer_xFthlAlph}
\vspace{-0.7cm}
\end{figure}
\subsection{Comparison of Cache Placement Strategies} \label{compcacheplace}
In addition to the \emph{popularity aware} cache placement scheme considered throughout the simulations, in this subsection, we consider three other caching strategies adopted from \cite{7925639}. The \emph{conservative} scheme stores the same files at every BS, i.e., the most popular files are cached at every BS by respecting the cache size constraints. In contrast, using the \emph{minimum redundancy} strategy, the BSs follow a specific pattern that aims at reducing redundancy by storing different contents. Also, another scheme where BSs store files randomly, independent from each other, is denoted as \emph{random} caching scheme in this subsection. \\
\indent Fig. \ref{yEE_xAlph_10cach} shows the EE as a function of the Zipf parameter $a$ for the four considered cache placement strategies. Different to previous considerations, we set $P_b^{\text{proc}} = 15$ dBm, $P_b^{\text{fthl}} = 9$ dBm and $F_{b,c} = 40$ Mbps, and we set the cache size to $10$ files. First, we observe that the EE of the \emph{conservative} and \emph{popularity aware} caching scheme is increasing with $a$. As $a$ becomes larger, some contents become more popular. In contrast, low values of $a$ imply that the content's popularity is almost uniformly distributed. Therefore, \emph{conservative} caching is particularly effective when $a$ becomes large, and \emph{popularity aware} caching behaves similarly (with a small portion of randomness \cite{7925639}). Under the \emph{minimum redundancy} caching scheme, the EE decreases as $a$ increases, i.e., only a few files are very popular, while the rest of the files experiences low popularity. This is due to the fact that using such scheme, all files, including the less popular ones, are cached. \emph{Random} caching performance experiences minor losses with increasing $a$. From Fig. \ref{yEE_xAlph_10cach}, we conclude that schemes which utilize knowledge of system parameters (i.e., $a$), e.g., \emph{conservative} and \emph{popularity aware} caching, are more beneficial from an EE perspective than more primitive schemes, and are thus better adopted in the context of our paper.
\subsection{Impact of Processing Power}
Before describing the next set of simulations, we introduce another state-of-the-art scheme, which is used as a baseline in our next set of simulations. In \cite{7488289}, the authors propose fixing the optimization variable $\mathbf{\tilde{t}}$ a priori, as cache placement and user requests are known, before jointly optimizing serving clusters and beamforming vectors. A BS that caches the requested file for user $k$ has to serve this user and process its data locally. This method leaves no choice for a BS to leave the computation to the respective CP. This is motivated by the fact, that from an EE perspective, local processing should be preferred to cloud computing since energy usage is lower in some regimes. This state-of-the-art scheme is referred to as Forced Local Computation.\\ 

In Fig. \ref{yEE_xFthl_fixtpl}, the EE as a function of processing power $P_k^\text{proc}$ is shown for two different cache sizes, where the fronthaul capacity is set to $40$ Mbps. In these simulations, our proposed scheme achieves better EE, when considering the required processing power for local caches. 
Interestingly, focusing on our proposed scheme, we notice a convergence of the EE for the two cache sizes, when $P_k^{\text{proc}}$ increases. This matches the observations from \ref{pprocvscach}. 
In fact, in networks where users require computationally intensive services, caching may not be helpful as an EE metric. The main cause for such behavior comes from the EE metric itself, since both BSs and CPs have to allocate $P_k^{\text{proc}}$ when serving user $k$. As this value increases it becomes more dominant over the fronthaul processing cost $P_k^{\text{fthl}}$, which is only applied when a CP processes user $k$'s data. From an EE perspective, the outsourcing of computation to the BSs becomes less significant with increasing processing costs.
\begin{figure}
	\centering
	\begin{subfigure}[b]{0.49\textwidth}
		\includegraphics[width=\linewidth]{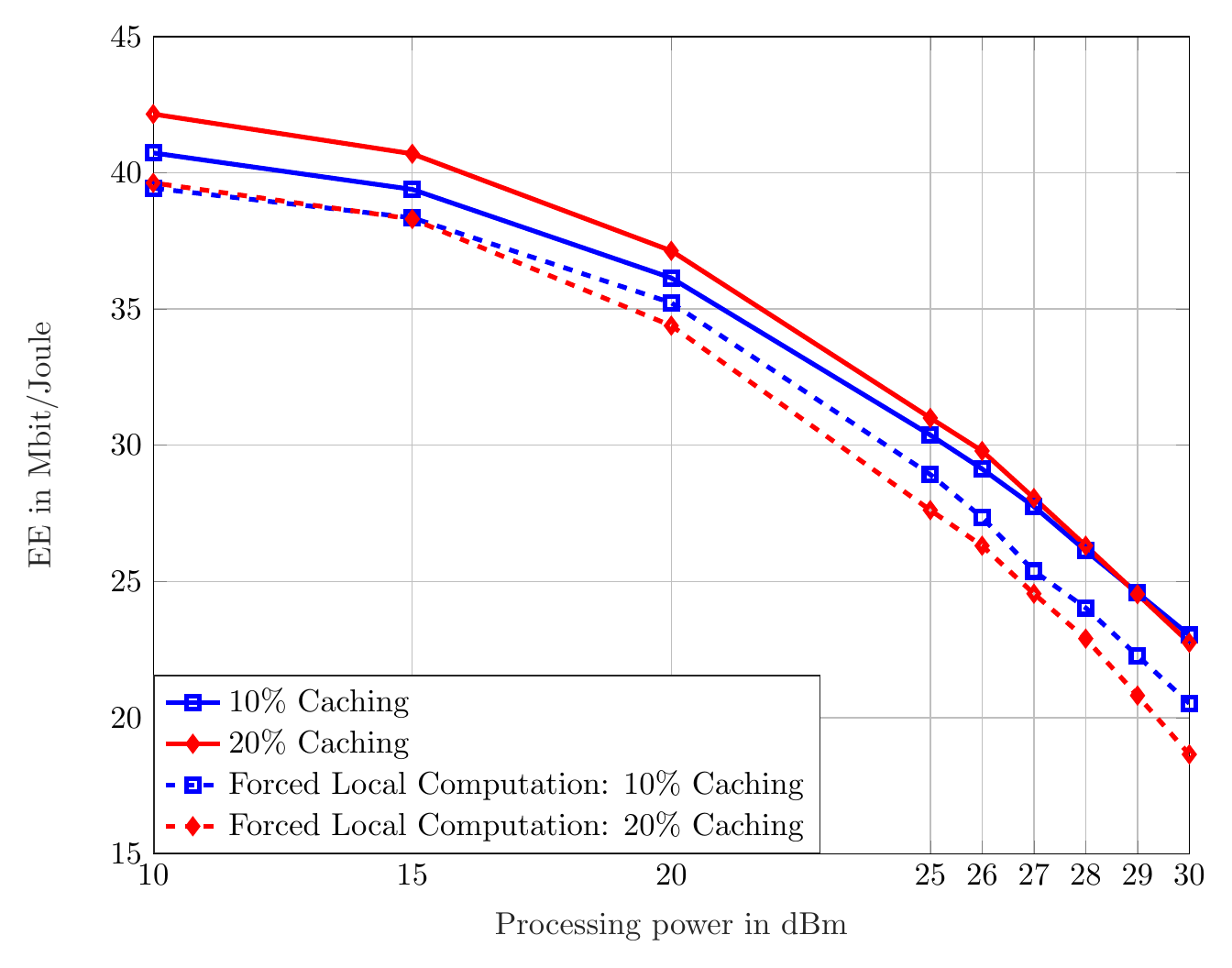}
		\caption{A comparison with a state-of-the-art.}
		\label{yEE_xFthl_fixtpl}
	\end{subfigure}\hfill
	\begin{subfigure}[b]{0.49\textwidth}
		\includegraphics[width=\linewidth]{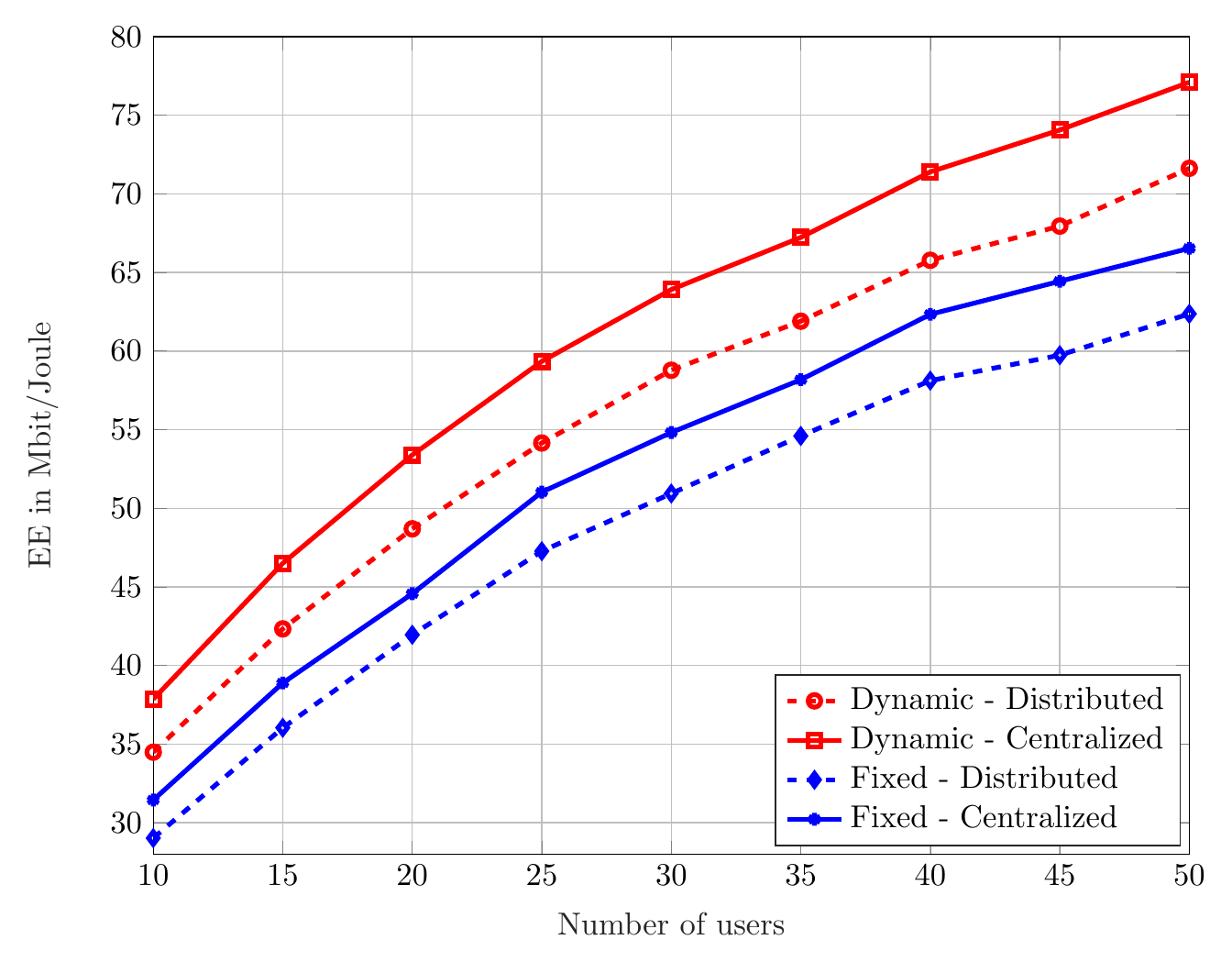}
		\caption{Cache of $10$ files.}
		\label{yEE_xNumU_xcach}
	\end{subfigure}
	\caption{EE as a function of fronthaul capacity and number of users for different cache sizes and schemes.}
	\label{yEE_xFthl_xNumU}
	\vspace{-0.7cm}
\end{figure}
\subsection{Impact of Network Densification}
In Fig. \ref{yEE_xNumU_xcach}, we examine the EE as a function of the number of users. 
Consider a network of $14$ BSs, where the processing power is set to $P_k^{\text{proc}} = 10$ dBm and the fronthaul capacity is $80$ Mbps. As a first observation, we find, that generally with more users the EE increases.
Similar to previous observations, we observe that the dynamic implementation always outperforms the fixed association implementation for all schemes, which highlights the numerical aspect of the clustering approach proposed in our paper.
\subsection{Convergence Behavior}\label{ssec:Conv}
\begin{figure}
	\begin{center}
		\includegraphics[scale=.58]{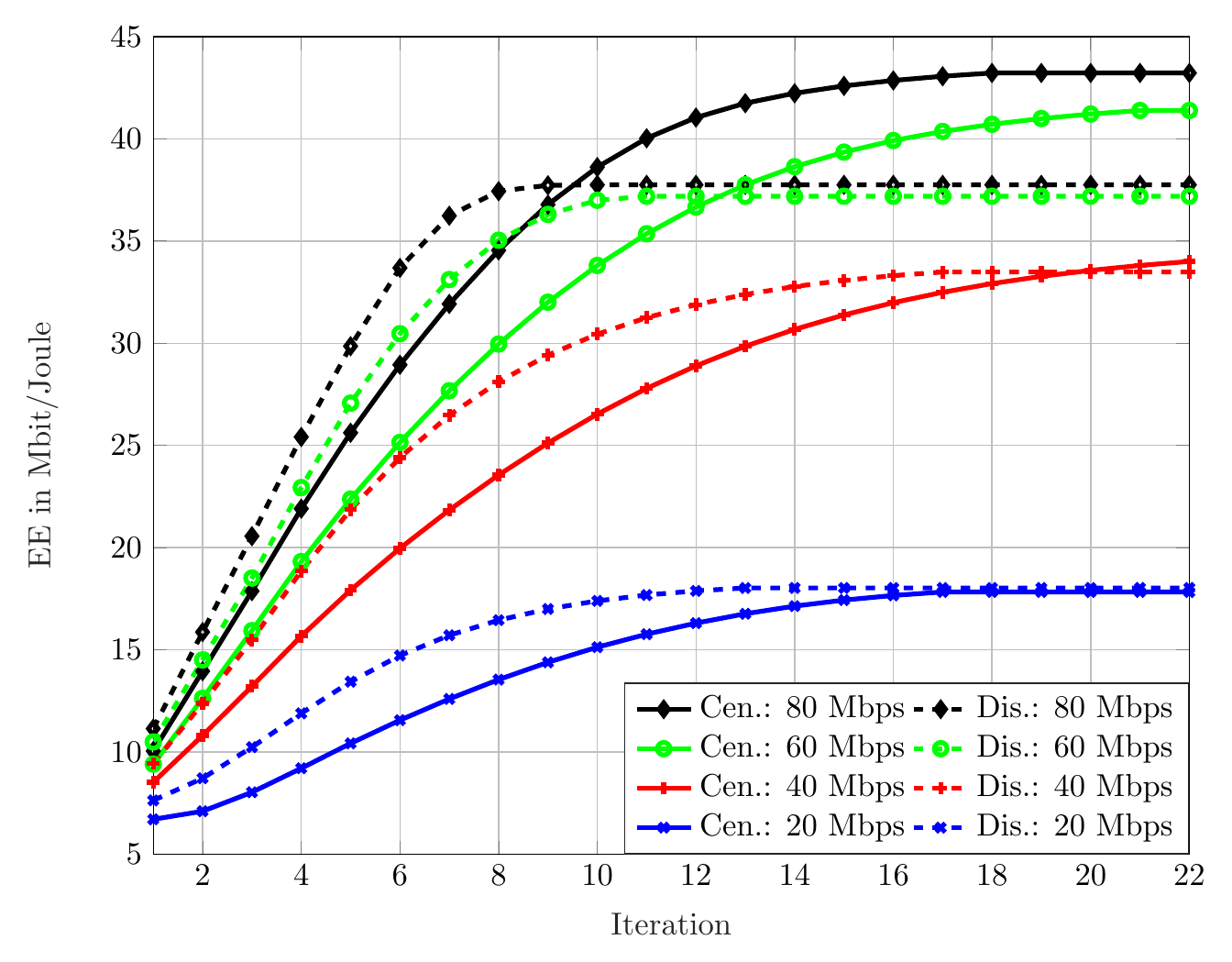}
		\caption{Convergence behavior of the distributed and centralized implementation of our proposed algorithm.} \label{convergence_tz23_20percent_alternative}
	\end{center}
	\vspace{-0.7cm}
\end{figure}
In another set of simulations, the processing power $P_k^{\text{proc}}$ is set to $10\;$dBm and the BSs can cache up to $20$ files. We only consider our proposed dynamic clustering scheme. 
To now illustrate the convergence behavior of our proposed Algorithm \ref{alg:sca_dinkel}, we plot the objective as a function of the number of iterations executed until converging in Fig. \ref{convergence_tz23_20percent_alternative}. We compare a distributed as well as a centralized implementation for different fronthaul capacities, i.e., $F_{b,c} \in \{20,40,60,80\}$ Mbps. In Fig. \ref{convergence_tz23_20percent_alternative}, the advantages of our algorithm in terms of convergence behavior and execution speed are highlighted, as the maximum iterations required for convergence are comparatively low. In all cases, the centralized implementation takes more iterations until convergence as compared to the distributed implementation. Fig. \ref{convergence_tz23_20percent_alternative} also shows that, at low fronthaul capacities, the distributed algorithm has acceptable loss in terms of EE against a centralized approach and, at the same time, performs better in terms of convergence, which constitutes an additional numerical feature of our proposed distributed resource allocation framework.

\subsection{Comparison with Sum-Rate Maximization Algorithm}\label{comp}
Finally, we make a connection to the weighted sum-rate maximization in MC-RAN from our conference version \cite{ICC}, which solves a mixed discrete-continuous non-convex optimization problem using a different fractional programming approach to tackle the non-convex part and $l_0$-norm approximation for the binary association part. Towards that end, \cite{ICC} proposes an efficient iterative algorithm for joint association and beamformer design that can be implemented in a distributed fashion across multiple CPs. 
To conduct a fair comparison, we now examine a special case of Algorithm \ref{alg:sca_dinkel}, i.e., the case when there are no cache hits ($\mathcal{K}_1 = \emptyset$ and $\mathcal{K}_2 = \mathcal{K}$). This is due to the fact, that we do not consider edge intelligence and caches in \cite{ICC}. 
We fix the processing power as $P_k^{\text{proc}} = 10$ dBm and the number of BSs as $14$. To tackle the problem of comparing a maximized sum-rate to a maximized EE, we propose converting the results of the algorithm from \cite{ICC} into the EE metric. Therefore we parse the final optimization variables from the sum-rate maximization into the EE formulation in \eqref{eq:f3}. Such comparison is done for both centralized and distributed implementations in Fig. \ref{yEE_xFthl_comp}.\\
Fig. \ref{yEE_xFthl_comp} shows that, for every fronthaul capacity, both Dinkelbach-like implementations, referring to Algorithm \ref{alg:sca_dinkel}, outperform the sum-rate maximization (SRM) implementation. Interestingly, we see a difference in the gain of using a centralized over a distributed implementation in the SRM case. Different from previous observations about our proposed scheme, the loss of the distributed SRM implementation in terms of EE is vastly increased. 
\begin{figure}
	\begin{center}
		\includegraphics[scale=.58]{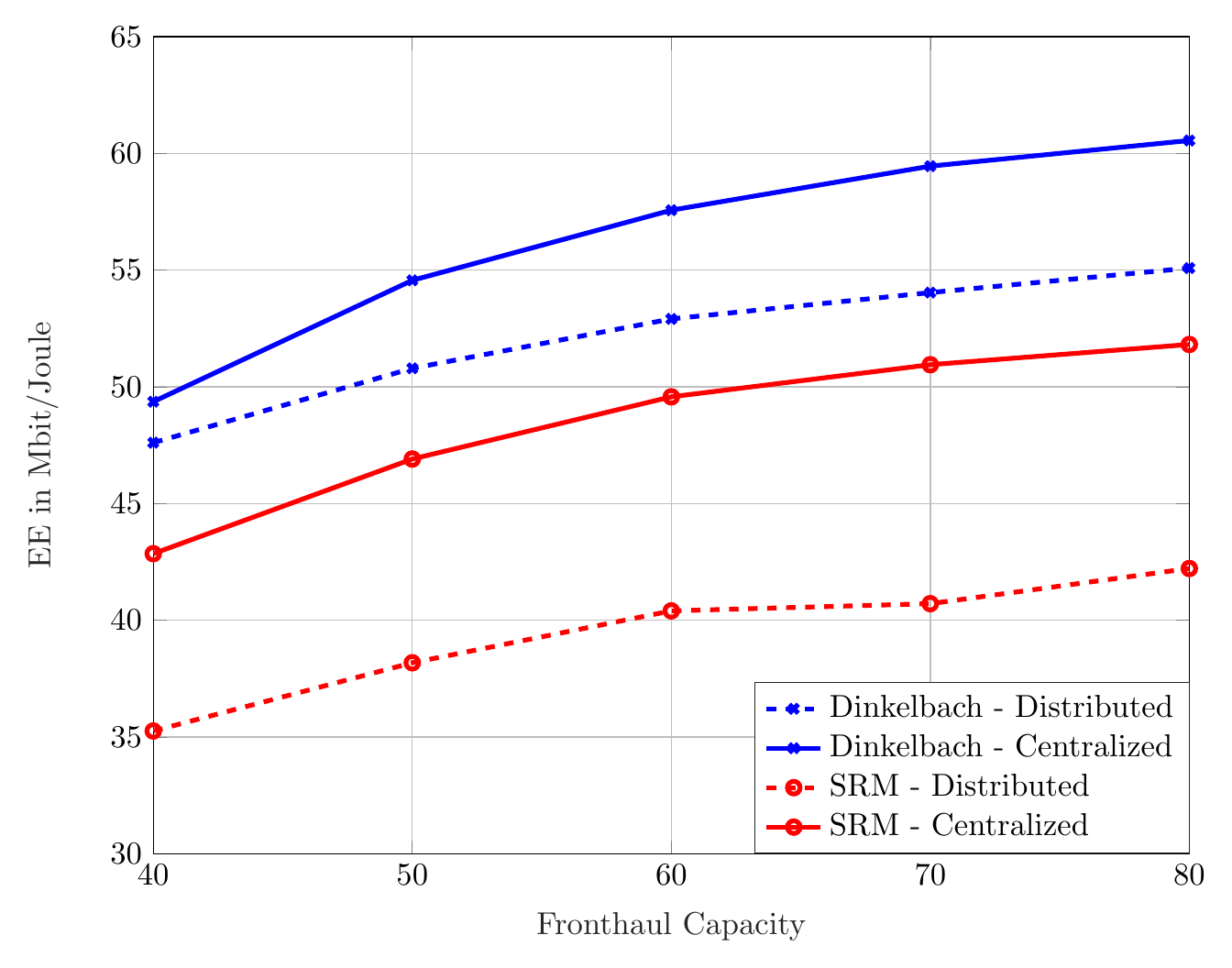}
		\caption{EE as a function of fronthaul capacity comparing the proposed algorithm and the algorithm from \cite{ICC}. We use the abbreviation SRM for sum-rate maximization.} \label{yEE_xFthl_comp}
	\end{center}
	\vspace{-0.7cm}
\end{figure}

\section{Conclusion}
Managing wireless systems with multiple CPs is a promising technique to cope with B5G network requirements. This paper considers an MC-RAN, where each cloud is connected to a distinct set of BSs via limited capacity fronthaul links. The BSs are equipped with local cache storage and baseband processing capabilities, as a means to alleviate the fronthaul congestion problem. The paper then investigates the problem of jointly assigning users to clouds and determining their beamforming vectors so as to maximize the network-wide energy efficiency subject to fronthaul capacity and transmit power constraints. This paper solves such a mixed discrete-continuous, non-convex optimization problem using fractional programming and successive inner-convex approximation techniques to deal with the non-convexity of the continuous part of the problem, and $l_0$-norm approximation to account for the binary association part. A highlight of the proposed algorithm is its capability of being implemented in a distributed fashion across the network's multiple clouds through a reasonable amount of information exchange. The numerical simulations illustrate the pronounced role the proposed algorithm plays in alleviating the interference of large-scale MC-RANs, especially in dense networks.	

\appendices

\section{Proof of lemma \ref{lma1}} \label{app1}
	First we revisit the function \eqref{DC} and define the difference of convex functions
	\begin{equation}
		\vartheta(\mathbf{x}) \triangleq \sum_{k\in\mathcal{K}}\frac{1}{4}\big( \underbrace{\left(t_{k,b}+R_{c,k}\right)^2}_{\vartheta^+(\mathbf{x})} - \underbrace{\left(t_{k,b} - R_{c,k}\right)^2}_{\vartheta^-(\mathbf{x})}\big)-F_{b,c},
	\end{equation}
	where the functions $\vartheta^+(\mathbf{x})$ and $\vartheta^-(\mathbf{x})$ are convex and $\mathbf{x}=\left[\mathbf{t}^T,\mathbf{r}^T\right]^T$. Now we only replace the concave part $-\vartheta^-(\mathbf{x})$ by its first order approximation around point $\left(t'_{k,b},R'_{c,k}\right)$ and thus get the convex upper approximation of function $\vartheta(\mathbf{x})$:
	\begin{equation}
		\hat{\vartheta}(\mathbf{x},\mathbf{x}') \triangleq \sum_{k\in\mathcal{K}}\frac{1}{4} \left(\vartheta^+(\mathbf{x}) - \vartheta^-(\mathbf{x}') - \nabla_\mathbf{x}\vartheta^-(\mathbf{x}')^T\left(\mathbf{x}-\mathbf{x}'\right)\right)-F_{b,c}.
	\end{equation}
	Here we have $\mathbf{x}'=\left[\mathbf{t}'^T,\mathbf{r}'^T\right]^T$, and we can rewrite function $\hat{\vartheta}(\mathbf{x},\mathbf{x}')$ as
	\begin{align}
		\hat{\vartheta}(\mathbf{x},\mathbf{x}') \triangleq& \sum_{k\in\mathcal{K}}\frac{1}{4}\bigg( \left(t_{k,b}+R_{c,k}\right)^2 - \left(t'_{k,b} - R'_{c,k}\right)^2 \nonumber\\
		&- 2\left(t'_{k,b}-R'_{c,k}\right)\left(t_{k,b}-t'_{k,b}\right) + 2\left(t'_{k,b}-R'_{c,k}\right)\left(R_{k,b}-R'_{k,b}\right) \bigg)-F_{b,c}.
	\end{align}
	Since $\hat{\vartheta}(\mathbf{x},\mathbf{x}')$ is a convex upper approximation of $\vartheta(\mathbf{x})$ the following inequality is valid: $\vartheta(\mathbf{x}) \leq \hat{\vartheta}(\mathbf{x},\mathbf{x}')$. As $\hat{\vartheta}(\mathbf{x},\mathbf{x}')$ can be transformed into $g_1(\mathbf{t},\mathbf{r},\mathbf{t}',\mathbf{r}')$, and $\vartheta(\mathbf{x})$ corresponds to the right hand side of \eqref{eq:lma11}, this completes the proof.
\section{Proof of lemma \ref{lma2}} \label{app2}
As the functions $\nabla_{\mathbf{x}}\zeta(\mathbf{x},\xi) = \frac{2}{\xi}(\mathbf{x})^\dagger$ and $ \nabla_{\xi}\zeta(\mathbf{x},\xi) = - \frac{1}{(\xi)^2} |\mathbf{x}|^2 $ are partial derivatives of $\zeta(\mathbf{x},\xi)$, the first-order Taylor expansion is
\begin{align}
	\tilde{\zeta}(\mathbf{x},\xi,\mathbf{x}',\xi') &= \zeta(\mathbf{x}',\xi') + \nabla_{\mathbf{x}}\zeta(\mathbf{x}',\xi') (\mathbf{x}-\mathbf{x}') + \nabla_{\xi}\zeta(\mathbf{x}',\xi') (\xi - \xi')	\nonumber\\
	&= \frac{|\mathbf{x}'|^2}{\xi'} + \frac{2}{\xi'}(\mathbf{x}')^\dagger (\mathbf{x}-\mathbf{x}') - \frac{1}{(\xi')^2} |\mathbf{x'}|^2 (\xi - \xi') \nonumber\\
	&= \frac{|\mathbf{x}'|^2}{\xi'} + \frac{2}{\xi'}(\mathbf{x}')^\dagger \mathbf{x} - \frac{2}{\xi'}(\mathbf{x}')^\dagger \mathbf{x}' - \frac{\xi}{(\xi')^2} |\mathbf{x'}|^2  + \frac{\xi'}{(\xi')^2} |\mathbf{x'}|^2 \nonumber\\
	&= \frac{1}{\xi'} |\mathbf{x}'|^2 + \frac{1}{\xi'} |\mathbf{x'}|^2 - \frac{2}{\xi'} |\mathbf{x}'|^2 + \frac{2}{\xi'}(\mathbf{x}')^\dagger \mathbf{x} - \frac{\xi}{(\xi')^2} |\mathbf{x'}|^2 \nonumber\\
	&= \frac{2}{\xi'}\Re\big\{ \left(\mathbf{x}'\right)^\dagger \mathbf{x} \big\} - \frac{\xi}{(\xi')^2} |\mathbf{x'}|^2.
\end{align}
This completes the proof.

\section{Proof of theorem \ref{theorem}} \label{app3}
The proof of Theorem \ref{theorem} consists of two parts. First, we show that the inner loop, i.e., the Dinkelbach algorithm, of Algorithm \ref{alg:sca_dinkel} converges to the global optimal solution of problem \eqref{eq:Opt5} in each step of the outer loop for the distributed case, assuming fixed inter-cloud interference. In the second step, the outer loop's convergence of Algorithm \ref{alg:sca_dinkel} to a stationary solution of problem \eqref{eq:Opt4} is shown.\\
\indent As a preliminary, the distributed problem \eqref{eq:Opt4} at cloud $c$ becomes
\begin{subequations}\label{eq:Opt4_dis}
	\begin{align}
		&\underset{\mathbf{w}_c, \mathbf{\tilde{w}}_c, \bm {\gamma}_c,\mathbf{r}_c}{\text{maximize}}\quad f_{2,\text{EE}}(c) \label{eq:Obj4_dis} \\
		&\text{subject to} \quad \eqref{eq:BM}, \eqref{eq:powe}, \eqref{eq:FrontC}, \eqref{eq:BgM}, \eqref{eq:BgM2}, \eqref{eq:q1}, \eqref{eq:q2}, \eqref{eq:BME4}, \eqref{eq:FrontC4}, \nonumber
	\end{align}
\end{subequations}
where the constraints are evaluated at cloud $c$ only and all other clouds $c' \neq c$ contribute fixed interference. Similarly, problem \eqref{eq:Opt5} becomes
\begin{subequations}\label{eq:Opt5_dis}
	\begin{align}
		&\underset{\mathbf{y}_c}{\text{maximize}}\quad f_{3,\text{EE}}(c) \label{eq:Obj5_dis} \\
		&\text{subject to} \quad \eqref{eq:BM}, \eqref{eq:powe}, \eqref{eq:BgM}, \eqref{eq:BgM2}, \eqref{eq:BME4}, \eqref{eq:beta1}, \eqref{eq:beta2}, \eqref{eq:beta3}, \nonumber\\
		&g_1(\mathbf{t}_c,\mathbf{r}_c;\mathbf{t}_c',\mathbf{r}_c') \leq 0 &&\forall b\in\mathcal{B}_c, \label{eq:g1_dis}\\
		&g_2({\bm \gamma}_c,\mathbf{r}_c;{\bm \gamma}_c') \leq 0 &&\forall k\in\mathcal{K},\label{eq:g2_dis}\\
		&g_3(\mathbf{w}_c,\tilde{\mathbf{w}}_c,{\bm \gamma}_c;\mathbf{w}_c',\tilde{\mathbf{w}}_c',{\bm \gamma}_c') \leq 0 &&\forall k\in\mathcal{K},\label{eq:g3_dis}
	\end{align}
\end{subequations}
where the optimization is over $\mathbf{y}_c$ defined as $\mathbf{y}_c=\left[ \mathbf{w}_c^T, \mathbf{\tilde{w}}_c^T, \mathbf{t}_c^T, \tilde{\mathbf{t}}_c^T, \mathbf{u}_c^T, \bm {\gamma}_c^T, \mathbf{r}_c^T \right]^T$ containing all optimization variables. The auxiliary convex optimization problem for the Dinkelbach algorithm from \eqref{eq:Fcl}, now takes the form of
\begin{equation}\label{eq:Fcl_dis}
	F({\lambda_{j}}) = \underset{\mathbf{y}_c\in\mathcal{Y}_c}{\text{max}} \left\{ g_{4}(c) - \lambda_{j}(c) g_{5}(c) \right\},
\end{equation}
where $\mathcal{Y}_c$ is the convex feasible set from problem \eqref{eq:Opt5_dis}.
\\
\indent The steps of showing the convergence of the Dinkelbach algorithm to a global optimal solution of \eqref{eq:Opt5} follow \cite[Proposition 3.2]{8187084} and are omitted here for brevity.\\
\indent Note that all following discussions are made under the assumption of investigating a single cloud $c$, whereas the other clouds keep the interference fixed.
The approximations in \eqref{eq:FrontRe2}, \eqref{eq:q1Re} and \eqref{eq:g3_} satisfy the conditions in \cite[Section III.c, Assumptions 1-3]{7862919}. To prove this, we shift our focus on \eqref{eq:FrontRe2}. We define
\begin{equation}
	\tilde{g}_1(\mathbf{x}) \triangleq \sum_{k \in \mathcal{K}}t_{k,b} \, r_{c,k} - F_{b,c}, \label{tildeg1_theorem1}
\end{equation}
and $g_1(\mathbf{x},\mathbf{x}') \triangleq g_1(\mathbf{t}_c,\mathbf{r}_c,\mathbf{t}_c',\mathbf{r}_c')$, where $\mathbf{x} = [\mathbf{t}_c^T,\mathbf{r}_c^T]^T$ and $\mathbf{x}' = [\mathbf{t}_c'^T,\mathbf{r}_c'^T]^T$. Towards this end, we show that following properties are satisfied:
\begin{itemize}
	\item[T1)] $g_1(\mathbf{x}',\mathbf{x}') = \tilde{g}_1(\mathbf{x}')$.
	\item[T2)] $g_1(\mathbf{x},\mathbf{x}') \geq \tilde{g}_1(\mathbf{x}), \qquad \forall\mathbf{x}'\in \mathcal{Z}$.
	\item[T3)] $g_1(\bullet,\mathbf{x}') \text{ is a convex function}, \qquad \forall\mathbf{x}'\in \mathcal{Z}$.
	\item[T4)] $g_1(\bullet,\bullet) \text{ is a continuous function on the feasible set.}$
	\item[T5)] $\nabla_{\mathbf{x}}g_1(\mathbf{x}',\mathbf{x}') = \nabla_{\mathbf{x}}\tilde{g}_1(\mathbf{x}')$.
	\item[T6)] $\text{The function } \nabla_{\mathbf{x}}g_1(\bullet,\bullet) \text{ is continuous on the feasible set.}$
\end{itemize}
Injecting $\mathbf{x}'$ into \eqref{eq:FrontRe2} and \eqref{tildeg1_theorem1} ensures the equality in T1
\begin{align}
	g_1(\mathbf{x}',\mathbf{x}') &= \tilde{g}_1(\mathbf{x}') \nonumber \\
	\sum_{k \in \mathcal{K}}\big( \left(t'_{k,b}+r'_{c,k}\right)^2 - 2 \left(t'_{k,b}-r'_{c,k}\right)\left(t'_{k,b}-r'_{c,k}\right) + &\nonumber\\
	\left(t'_{k,b}-r'_{c,k}\right)^2\big) - 4 F_{b,c} &= \sum_{k \in \mathcal{K}}t'_{k,b} \, r'_{c,k} - F_{b,c} \qquad\qquad \nonumber \\
	\frac{1}{4} \left( \left(t'_{k,b}+r'_{c,k}\right)^2 - \left(t'_{k,b}-r'_{c,k}\right)^2 \right) &= t'_{k,b} r'_{c,k} \nonumber \\
	\frac{1}{4} \left( 4 t'_{k,b}\, r'_{c,k}\right) &= t'_{k,b} \, r'_{c,k} \nonumber .
\end{align}
T2 follows directly from Lemma~\ref{lma1}. To verify T3 and T4, we take a closer look at the structure of $g_1(\bullet,\mathbf{x}')$ with
\begin{equation}
	g_1(\mathbf{x},\mathbf{x}') = \sum_{k \in \mathcal{K}}\big( \underbrace{\left(t_{k,b}+r_{c,k}\right)^2}_{\text{convex}} - \underbrace{2 \left(t'_{k,b}-r'_{c,k}\right)\left(t_{k,b}-r_{c,k}\right)}_{\text{linear}} + \left(t'_{k,b}-r'_{c,k}\right)^2\big) - 4 F_{b,c}.
\end{equation}
As $g_1(\bullet,\mathbf{x}')$, with fixed $\mathbf{x}'$, consists of a convex quadratic function subtracted by a linear function, which is convex, T3 holds. $g_1(\bullet,\bullet)$ has convex quadratic and bilinear terms, thus T4 is satisfied.
For T5 and T6 we take the partial derivatives as follows
\begin{equation}
	\nabla_{\mathbf{x}}\tilde{g}_1(\mathbf{x})
	\begin{cases}
		\frac{\partial \tilde{g}_1(\mathbf{x})}{\partial \mathbf{t}} = \sum_{k \in \mathcal{K}}r_{c,k}\\
		\frac{\partial \tilde{g}_1(\mathbf{x})}{\partial \mathbf{r}} = \sum_{k \in \mathcal{K}}t_{k,b}
	\end{cases}; \label{nabla1}
\end{equation}
\begin{equation}
	\nabla_{\mathbf{x}}g_1(\mathbf{x},\mathbf{x}')
	\begin{cases}
		\frac{\partial g_1(\mathbf{x},\mathbf{x}')}{\partial \mathbf{t}} = \sum_{k \in \mathcal{K}}\frac{1}{4}\left( 2\left(t_{k,b}+r_{c,k}\right) - 2\left(t'_{k,b}-r'_{c,k}\right) \right)\\
		\frac{\partial g_1(\mathbf{x},\mathbf{x}')}{\partial \mathbf{r}} = \sum_{k \in \mathcal{K}}\frac{1}{4}\left( 2\left(t_{k,b}+r_{c,k}\right) + 2\left(t'_{k,b}-r'_{c,k}\right) \right)
	\end{cases}. \label{nabla2}
\end{equation}
The substitution of $\mathbf{x}$ with $\mathbf{x}'$ in \eqref{nabla1} and \eqref{nabla2} is straightforward and yields the equality T5. Since $\nabla_{\mathbf{x}}g_1(\bullet,\bullet)$ consists of linear terms, T6 holds. Similarly, these steps can be followed to verify T1-T6 for the remaining two reformulations in \eqref{eq:q1Re} and \eqref{eq:g3_}. This completes the proof.

\bibliographystyle{IEEEtran}
\bibliography{bibliography}
\balance
\end{document}